%% file: article.tex
\documentclass[a4paper,UKenglish,cleveref, autoref, thm-restate]{lipics-v2021}

\usepackage{mathtools}

\usepackage{tikz}
\usepackage{tikz-qtree}

\renewcommand{\arraystretch}{1.8}

\newcommand{\lBrace}{\{\!\!\{}
\newcommand{\rBrace}{\}\!\!\}}

\title{The Trie Measure, Revisited} 

\author{Jarno N. {Alanko}}{University of Helsinki, Finland}{jarno.alanko@helsinki.fi }{https://orcid.org/0000-0002-8003-9225}{}

\author{Ruben {Becker}}{Ca' Foscari University of Venice, Italy}{rubensimon.becker@unive.it}{https://orcid.org/0000-0002-3495-3753}{}
\author{Davide {Cenzato}}{Ca' Foscari University of Venice, Italy}{davide.cenzato@unive.it}{https://orcid.org/0000-0002-0098-3620}{}
\author{Travis {Gagie}}{Dalhousie University, Halifax, Nova Scotia, Canada}{Travis.Gagie@dal.ca }{https://orcid.org/0000-0003-3689-327X}{}
\author{Sung-Hwan {Kim}}{Ca' Foscari University of Venice, Italy}{sunghwan.kim@unive.it}{https://orcid.org/0000-0002-1117-5020}{}
\author{Bojana {Kodric}}{Ca' Foscari University of Venice, Italy}{bojana.kodric@unive.it}{https://orcid.org/0000-0001-7242-0096}{}
\author{Nicola {Prezza}}{Ca' Foscari University of Venice, Italy}{nicola.prezza@unive.it}{https://orcid.org/0000-0003-3553-4953}{}

\funding{\textit{Ruben Becker, Davide Cenzato, Sung-Hwan Kim, Bojana Kodric, Nicola Prezza}: Funded by the European Union (ERC, REGINDEX, 101039208). 
Views and opinions expressed are however those of the author(s) only and do not necessarily reflect those of the European Union or the European Research Council Executive Agency. Neither the European Union nor the granting authority can be held responsible for them.
\textit{Travis Gagie}: Funded by NSERC Discovery Grant RGPIN-07185-2020.}

\authorrunning{J. Alanko et al.} 

\Copyright{Jarno Alanko et al.} 



\category{} 

\relatedversion{} 

\supplement{}
\supplementdetails[subcategory={}, cite={}, swhid={}]{Software}{https://github.com/regindex/trie-measure}

\nolinenumbers 

\EventEditors{Paola Bonizzoni and Veli M\"{a}kinen}
\EventNoEds{2}
\EventLongTitle{36th Annual Symposium on Combinatorial Pattern Matching (CPM 2025)}
\EventShortTitle{CPM 2025}
\EventAcronym{CPM}
\EventYear{2025}
\EventDate{June 17--19, 2025}
\EventLocation{Milan, Italy}
\EventLogo{}
\SeriesVolume{331}
\ArticleNo{20}

\input{macros}

\begin{document}

\maketitle

\begin{abstract}
In this paper, we study the following problem: given $n$ subsets $S_1, \dots, S_n$ of an integer universe $U = \{0,\dots, u-1\}$, having total cardinality $N = \sum_{i=1}^n |S_i|$, find a prefix-free encoding $\enc : U \rightarrow \{0,1\}^+$ minimizing the so-called \emph{trie measure}, i.e., the total number of edges in the $n$ binary tries $\mathcal T_1, \dots, \mathcal T_n$, where $\mathcal T_i$ is the trie packing the encoded integers $\{\enc(x):x\in S_i\}$. 
We first observe that this problem is equivalent to that of merging $u$ sets with the cheapest sequence of binary unions, a problem which in [Ghosh et al., ICDCS 2015] is shown to be NP-hard.
Motivated by the hardness of the general problem, we focus on particular families of prefix-free encodings. We start by studying the fixed-length \emph{shifted encoding} of [Gupta et al., Theoretical Computer Science 2007]. Given a parameter $0\le a < u$, this encoding sends each $x \in U$ to $(x + a) \modulo u$, interpreted as a bit-string of $\log u$ bits. We develop the first efficient algorithms that find the value of $a$ minimizing the trie measure when this encoding is used. Our two algorithms run in $O(u + N\log u)$ and $O(N\log^2 u)$ time, respectively. 
We proceed by studying \emph{ordered  encodings} (a.k.a. \emph{monotone} or \emph{alphabetic}), and describe an algorithm finding the optimal such encoding in $O(N+u^3)$ time.
Within the same running time, we show how to compute the best \emph{shifted ordered encoding}, provably no worse than \emph{both} the optimal shifted and optimal ordered encodings. 
We provide implementations of our algorithms and discuss how these encodings perform in practice. 
\end{abstract}

\clearpage
\newpage

\section{Introduction}

Consider the problem of encoding a set of integers $S\subseteq U = \{0,\dots, u-1\}$ (without loss of generality, we assume $u$ to be a power of two), so as to minimize the overall number of bits used to represent $S$.
In their seminal work on data-aware measures, Gupta et al.~\cite{GUPTA2007313} proposed and analyzed an encoding for sets of integers based on the idea of packing the integers (seen as strings of $\log u$ bits) into a binary trie: this allows to avoid storing multiple times shared prefixes among the encodings of the integers. 
The number of edges in such a trie is known as the \emph{trie measure} of the set. See Figure \ref{fig:trie-example1} for an example.

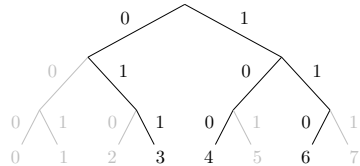
\begin{figure}[h!]
\centering
\begin{tikzpicture}[level distance=1cm,sibling distance=.5cm,
   edge from parent path={(\tikzparentnode) -- (\tikzchildnode)}, scale=.7, transform shape]
\Tree
[
    \edge node[auto=right,pos=.5] {$0$};
    [ 
        \edge [lightgray] node[auto=right,pos=.5] {$0$};
        [
            \edge [lightgray] node[auto=right,pos=.7] {$0$}; [.\textcolor{lightgray}{0} ]
            \edge [lightgray] node[auto=left,pos=.7] {$1$}; [.\textcolor{lightgray}{1} ]
        ]
        \edge node[auto=left,pos=.5] {$1$};
        [
            \edge [lightgray] node[auto=right,pos=.7] {$0$}; [.\textcolor{lightgray}{2} ]
            \edge node[auto=left,pos=.7] {$1$}; [.3 ]
        ]
    ]
    \edge node[auto=left,pos=.5] {$1$};
    [ 
        \edge node[auto=right,pos=.5] {$0$};
        [
            \edge node[auto=right,pos=.7] {$0$}; [.4 ]
            \edge [lightgray] node[auto=left,pos=.7] {$1$}; [.\textcolor{lightgray}{5} ]
        ]
        \edge node[auto=left,pos=.5] {$1$};
        [
            \edge node[auto=right,pos=.7] {$0$}; [.6 ]
            \edge [lightgray] node[auto=left,pos=.7] {$1$}; [.\textcolor{lightgray}{7} ]
        ]
    ]
]
\end{tikzpicture}
\caption{\footnotesize Example of trie encoding the set of integers $\{3,4,6\} \subseteq \{0,1,\dots, 7\}$ over universe of size $u=8$. Black edges belong to the trie. Gray edges do not belong to the trie and are shown only for completeness. Each integer is encoded using $\log 8 = 3$ bits (logarithms are in base 2). The trie has 8 edges, so the \emph{trie measure} for this set using the standard integer encoding is 8.}
\label{fig:trie-example1}
\end{figure}

Gupta et al.\ also showed that this measure approaches worst-case entropy on expectation when the integers are shifted by a uniformly-random quantity modulo $u$, i.e.\ when building the trie over the set $S+a \coloneq \{x+a \modulo u: x\in S\}$ for a uniform $a\in U$. 
See Figure \ref{fig:trie-example2}. 

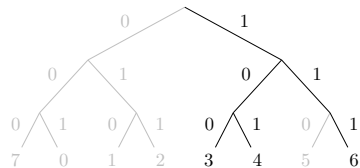
\begin{figure}[h!]
\centering
\begin{tikzpicture}[level distance=1cm,sibling distance=.5cm,
   edge from parent path={(\tikzparentnode) -- (\tikzchildnode)}, scale=.7, transform shape]
\Tree
[
    \edge [lightgray] node[auto=right,pos=.5] {$0$};
    [ 
        \edge [lightgray] node[auto=right,pos=.5] {$0$};
        [
            \edge [lightgray] node[auto=right,pos=.7] {$0$}; [.\textcolor{lightgray}{7} ]
            \edge [lightgray] node[auto=left,pos=.7] {$1$}; [.\textcolor{lightgray}{0} ]
        ]
        \edge [lightgray] node[auto=left,pos=.5] {$1$};
        [
            \edge [lightgray] node[auto=right,pos=.7] {$0$}; [.\textcolor{lightgray}{1} ]
            \edge [lightgray] node[auto=left,pos=.7] {$1$}; [.\textcolor{lightgray}{2} ]
        ]
    ]
    \edge node[auto=left,pos=.5] {$1$};
    [ 
        \edge node[auto=right,pos=.5] {$0$};
        [
            \edge node[auto=right,pos=.7] {$0$}; [.3 ]
            \edge node[auto=left,pos=.7] {$1$}; [.4 ]
        ]
        \edge node[auto=left,pos=.5] {$1$};
        [
            \edge [lightgray] node[auto=right,pos=.7] {$0$}; [.\textcolor{lightgray}{5} ]
            \edge node[auto=left,pos=.7] {$1$}; [.6 ]
        ]
    ]
]
\end{tikzpicture}
\caption{\footnotesize A trie that stores the same set of integers $\{3,4,6\} \subseteq \{0,1,\dots, 7\}$ of Figure \ref{fig:trie-example1}, but with the shifted integer encoding mapping each $x\in U$ to (the binary string of $\log u$ bits) $(x+1)\modulo 8$.
The trie has 6 edges, so the \emph{shifted trie measure} with shift $a=1$ is 6.}
\label{fig:trie-example2}
\end{figure}

In this paper, we revisit this problem in two natural directions: (1) we move from one set to a sequence of sets, and (2) we study more general prefix-free integer encodings. 
More formally, given $n$ subsets $S_1, \dots, S_n$ of $U$, having total cardinality $N = \sum_{i=1}^n |S_i|$, 
we study the problem of finding a prefix-free encoding $\enc : U \rightarrow \{0,1\}^+$ minimizing the \emph{trie measure} of these sets, i.e.\ the total number of edges in the $n$ binary tries $\mathcal T_1, \dots, \mathcal T_n$, where $\mathcal T_i$ is the trie packing the encoded integers $\{\enc(x):x\in S_i\}$.
Solving the problem on set sequences finds applications in data structures for rank/select queries on set sequences---also known in the literature as \emph{degenerate strings}---, an important toolbox in indexes for labeled graphs \cite{GAGIE201767,CotumaccioJACM2023,BOSS,XBWT}. Such applications stem from the recent work of Alanko et al.~\cite{AlankoBPV23}, who showed how to solve rank and select queries on set sequences with a data structure --- the \emph{subset wavelet tree} --- implicitly encoding sets as tries. While \cite{AlankoBPV23} only considered the standard binary (balanced) encoding, their technique works for any prefix-free encoding. As a result, our work can directly be applied to optimize the space of their data structure.
Another application is the offline set intersection problem where the goal to preprocess subsets of the universe, so that later given a set of indices of the sets, one can compute the intersection of the sets space-time efficiently. Arroyuelo and Castillo \cite{ArroyueloC23} presented an adaptive approach that uses a trie representation of sets and they showed that it can benefit from its space-efficient representation.

We begin by showing that
finding the optimal such prefix-free encoding is equivalent to the following natural optimization problem: find the minimum-cost sequence of binary unions that merges $u$ given sets, where the cost of a union is the sum of the two sets' cardinalities. As this problem was shown to be NP-hard by Ghosh et al.~\cite{Ghosh2015}, we focus on particular (hopefully easier) families of encodings. 
We start by studying the shifted encoding of Gupta et al.,
who did not discuss the problem of finding the optimal value of $a\in U$ minimizing the corresponding trie measure. We describe an algorithm solving the problem in $O(u + N\log u)$ time. The algorithm is based on the fact that, under this encoding, the trie measure has a highly periodic sub-structure as a function of the shift $a$.
We then remove the linear dependency on the universe size and achieve running time $O(N\log^2 u)$ by storing this periodic structure with a DAG. $O(N\log u)$ running time is also possible by merging the two ideas, but for space constraints we will describe it in the extended version of this article.
We also want to remark that the general prefix-free encoding problem is not only difficult to compute, but also it requires to store the ordering of the elements in addition to the trie representation. By a simple enumeration, it needs $O(N' \log u)$ bits of space where $N'$ is the cardinality of the union of the sets (i.e., the number of distinct elements over all sets). On the other hand, the shifted encoding only requires $O(\log u)$ bits since it is sufficient to store a single integer $a\in U$.

We then move to \emph{ordered} encodings: here, the encoding must preserve lexicographically the order of the universe $U$ (i.e.\ the standard integer order). In this case, we show that the textbook solution of Knuth \cite{knuth1971optimum} based on dynamic programming, running in $O(N+u^3)$-time, can be adapted to our scenario.  
We observe that essentially the same solution allows also shifting the universe (like in the shifted encoding discussed above) at no additional cost. As a result, we obtain that the best \emph{shifted ordered encoding} can be computed in  $O(N+u^3)$-time as well. This encoding is never worse than the best shifted and the best ordered encodings.

We conclude our paper with an experimental evaluation of our algorithms on real datasets, showing how these encodings perform with respect to the worst and average-case scenario.

\section{Preliminaries and Problem Formulation}
A bit-string is a finite sequence of bits, i.e., an element from $\{0,1\}^*$. 
We count indices from 1, so that for $\beta\in\{0,1\}^+$, $\beta[1]$ is the first element. 
With $\prec$ we denote the lexicographic order among bit-strings and with $|\cdot|$ we denote the length of bit-strings. For two strings $\alpha$ and $\beta$ of possibly different length, each not being a prefix of the other, we use $\oplus$ to denote the (non-commutative) operator defined as $\alpha \oplus \beta = \beta[j\ldots |\beta|]$, where $j-1$ is the length of the longest common prefix of $\alpha$ and $\beta$, i.e., $\alpha[i]=\beta[i]$ for $1 \le i < j$ and $\alpha[j]\neq \beta[j]$. 

We denote with $U:=\{0, 1, \ldots, u - 1\}$ the integer universe and we assume  that $u$ is a power of two. 
Logarithms are in base 2, so $\log(x)$ indicates $\log_2(x)$. 
In particular, $\log u$ is an integer.
Notation $[n]$ indicates the set $\{1,2,\dots, n\}$. For integers $\ell<r$, we denote with $[\ell,r)$ an interval of integers $\{\ell,\ell+1,\cdots,r-2,r-1\}$ and if $\ell=r$, then $[\ell,r)=\emptyset$ is the empty set. 
We use double curly braces $\lBrace \dots \rBrace$ to represent multisets.
Given a multiset 
$\mathcal{I}$ containing subsets of $U$ and an integer $a\in U$, the \emph{depth} of $\mathcal{I}$ at position $a$ is defined as the number of elements of $\mathcal{I}$ that contain $a$. 
For integer $a,b$ and $p$ with $p\ge 1$, we say $a\equiv_p b$ if and only if $a\modulo p=b\modulo p$.
For a predicate $P$, we use $\ones[P]$ to denote the indicator function that is 1 if $P$ holds and zero otherwise.

A prefix-free encoding of $U$ is a function $\enc: U \rightarrow \{0, 1\}^{+}$ that satisfies that for no two distinct $x, y\in U$, it holds that $\enc(x)$ is a prefix of $\enc(y)$.
For a set $S\subseteq U$, we let $\enc(S) := \{\enc(x):x\in S\}$ be the set of bitstrings obtained from encoding $S$ via $\enc$. We say that $\enc$ is \emph{ordered} if and only if $x<y$ implies $\enc(x) \prec \enc(y)$ for all $x,y\in U$.

We study the following \emph{trie measure}, generalized from the particular cases studied in \cite{GUPTA2007313}:

\begin{definition}\label{def:trie prefix-free}
    Let $\enc: U \rightarrow \{0,1\}^+$ be  a prefix-free encoding and 
    let $S = \{x_1, \ldots, x_m\} \subseteq U$ such that $i<j$ implies $\enc(x_{i}) \prec \enc(x_{j})$ in lexicographic order. We define
    \[
        \trie(\enc(S)) = |\enc(x_{i_1})| + 
        \sum_{j=2}^m |\enc(x_{i_{j-1}}) \oplus \enc(x_{i_j})|.
    \]
\end{definition}

In this article we work, more generally, with sequences of sets. The above definitions generalize naturally:

\begin{definition}\label{def:trie prefix-free sequence}
    Let $\enc: U \rightarrow \{0,1\}^+$ be  a prefix-free encoding and 
    let $\mathcal S = \langle S_1, \dots, S_n \rangle$ be a sequence of subsets of $U$. We define $\enc(\mathcal S)$ to be the sequence $\langle \enc(S_1), \dots, \enc(S_n) \rangle$, and:
    \[
        \trie(\enc(\mathcal S)) =
        \sum_{i=1}^n \trie(\enc(S_i))
    \]
\end{definition}

Throughout the article, we will denote with $N = \sum_{i=1}^n |S_i|$ the total cardinality of the input sets $S_1, \dots, S_n$. Definition \ref{def:trie prefix-free sequence} leads to the central problem explored in this paper, tackled in Sections \ref{section: encoding shifts} and \ref{sec:optimum ordered}: given a sequence $\mathcal S = \langle S_1, \dots, S_n \rangle$ of $n$ subsets of $U$, find the encoding $\enc$ minimizing $\trie(\enc(\mathcal S))$, possibly focusing on particular sub-classes of prefix-free encodings. 

A particularly interesting such sub-class, originally introduced by Gupta et al.~\cite{GUPTA2007313}, is that of \emph{shifted encodings}:

\begin{definition}\label{def:shifted encoding}
Let $S\subseteq U$ and $a\in U$. Notation $S+a$ denotes the application to $S$ of the prefix-free encoding sending each $x\in U$ to $(x+a)\modulo u$, interpreted as a bit-string of $\log u$ bits. Similarly, for a sequence $\mathcal S = \langle S_1, \dots, S_n \rangle$ of subsets of $U$  we denote by $\mathcal S+a$ the sequence $\langle S_1+a, \dots, S_n+a \rangle$.
\end{definition}

We call $\trie(\mathcal S + a)$ the \emph{shifted trie measure}. 
In Section  \ref{section: encoding shifts} we study the problem of finding the value $a\in U$ minimizing $\trie(\mathcal S + a)$, a problem left open by Gupta et al.~\cite{GUPTA2007313}.

\subsection{An Equivalent Problem Formulation}\label{sec:equivalent problem formulation}

Observe that our trie-encoding problem is equivalent to the following natural encoding problem, which we will resort to in Section \ref{sec:optimum ordered}. 
Let $\mathcal S = \langle S_1, \dots, S_n\rangle$ be a sequence of subsets of $U$. For each $x\in U$, denote with $A_x = \{i\in [n]\ :\ x\in S_i\}$ the set collecting all indices $i$ of the subsets $S_i$ containing $x$. Without loss of generality, in this reformulation we assume that $A_x \neq \emptyset$ for all $x\in U$ and $S_i\neq \emptyset$ for all $i\in [n]$ (otherwise, simply re-map integers and ignore empty sets $S_i$). We furthermore do not require $u$ to be a power of two (this will be strictly required only for the optimal shifted encoding in Section \ref{section: encoding shifts}).

For a given prefix-free encoding $\enc$, let $T^{\enc}$ be the binary trie storing the encodings $\enc(0), \dots, \enc(u-1)$ such that, for every $x\in U$, the leaf of $T^{\enc}$ reached by $\enc(x)$ is labeled with $x$. 
For any binary trie $T$ with leaves labeled by integers let moreover:  $r(T)$ be the root of $T$, $T_v$ denote the subtree of $T$ rooted at node $v$, and $L(T_v)$ be the set of (integers labeling the) leaves of subtree $T_v$ (i.e.\ the leaves below node $v$). 
We overload notation and identify with $T$ also the set of $T$'s nodes (the use will be always clear by the context).
Observe that the definition of the sets $A_x$ allows us to reformulate the measure $\trie(\enc(\mathcal S))$ as follows:
\[
    \trie(\enc(\mathcal S)) 
    = \sum_{i\in [n]} \trie(\enc(S_i))
    = \sum_{v\in T^{\enc}\setminus \{r(T^{\enc})\}} \Big| \bigcup_{x \in L(T^{\enc}_v)} A_x \Big|.
\]

In other words, the cost of node $v$ is equal to the cardinality of the union of the sets $A_{x_1}, \dots, A_{x_t}$ corresponding to the leaves $x_1,\dots,x_t$ below $v$. The overall cost of the tree is the sum of the costs of its nodes, excluding the root. 
To see why this formulation is equivalent to the previous one observe that, among the $n$ tries for $\enc(S_1), \dots, \enc(S_n)$, 
(a copy of) the incoming edge of $v$ is present in the trie for $\enc(S_i)$ for all 
$i \in \bigcup_{x \in L(T_v)} A_x$. See  Figure \ref{fig:optimum ordered}.

\begin{figure}[h!]
\centering
\begin{minipage}{.3\textwidth}
  \centering
\begin{tikzpicture}[level distance=1cm,sibling distance=.5cm,
   edge from parent path={(\tikzparentnode) -- (\tikzchildnode)}, scale=.7, transform shape]
\Tree
[
    \edge  node[auto=right,pos=.5] {$0$}; [.0 ]
    \edge node[auto=left,pos=.5] {$1$};
    [ 
        \edge node[auto=right,pos=.5] {$0$}; [.1 ]
        \edge [red] node[auto=left,pos=.5] {$1$};
        [
            \edge node[auto=right,pos=.7] {$0$}; [.2 ]
            \edge node[auto=left,pos=.7] {$1$}; [.3 ]
        ]
    ]
]
\end{tikzpicture}
\centering
\footnotesize $\begin{array}{c}U = \{0,1,2,3\}\end{array}$
\end{minipage}
\hfill\vline\hfill
\begin{minipage}{.2\textwidth}
\begin{tikzpicture}[level distance=1cm,sibling distance=.5cm,
   edge from parent path={(\tikzparentnode) -- (\tikzchildnode)}, scale=.7, transform shape]
\Tree
[
    \edge [lightgray] node[auto=right,pos=.5] {$0$}; [.\textcolor{lightgray}0 ]
    \edge node[auto=left,pos=.5] {$1$};
    [ 
        \edge node[auto=right,pos=.5] {$0$}; [.1 ]
        \edge [red] node[auto=left,pos=.5] {$1$};
        [
            \edge node[auto=right,pos=.7] {$0$}; [.2 ]
            \edge [lightgray] node[auto=left,pos=.7] {$1$}; [.\textcolor{lightgray}3 ]
        ]
    ]
]
\end{tikzpicture}
\centering
\footnotesize $\begin{array}{c}S_1 = \{1,2\}\end{array}$
\end{minipage}
\begin{minipage}{.2\textwidth}
\begin{tikzpicture}[level distance=1cm,sibling distance=.5cm,
   edge from parent path={(\tikzparentnode) -- (\tikzchildnode)}, scale=.7, transform shape]
\Tree
[
    \edge node[auto=right,pos=.5] {$0$}; [.0 ]
    \edge node[auto=left,pos=.5] {$1$};
    [ 
        \edge node[auto=right,pos=.5] {$0$}; [.1 ]
        \edge [lightgray] node[auto=left,pos=.5] {$1$};
        [
            \edge [lightgray] node[auto=right,pos=.7] {$0$}; [.\textcolor{lightgray}2 ]
            \edge [lightgray] node[auto=left,pos=.7] {$1$}; [.\textcolor{lightgray}3 ]
        ]
    ]
]
\end{tikzpicture}
\centering
\footnotesize $\begin{array}{c}S_2 = \{0,1\}\end{array}$
\end{minipage}
\begin{minipage}{.2\textwidth}
\begin{tikzpicture}[level distance=1cm,sibling distance=.5cm,
   edge from parent path={(\tikzparentnode) -- (\tikzchildnode)}, scale=.7, transform shape]
\Tree
[
    \edge [lightgray] node[auto=right,pos=.5] {$0$}; [.\textcolor{lightgray}0 ]
    \edge node[auto=left,pos=.5] {$1$};
    [ 
        \edge node[auto=right,pos=.5] {$0$}; [.1 ]
        \edge [red] node[auto=left,pos=.5] {$1$};
        [
            \edge node[auto=right,pos=.7] {$0$}; [.2 ]
            \edge node[auto=left,pos=.7] {$1$}; [.3 ]
        ]
    ]
];
\end{tikzpicture}
\centering
\footnotesize $\begin{array}{c}S_3 = \{1,2,3\}\end{array}$
\end{minipage}
\caption{\footnotesize \textbf{Left.} An ordered prefix-free encoding $\enc$ of the universe $U = \{0,1,2,3\}$, represented as a binary trie $T^{\enc}$. 
This encoding is used (right part) to encode sets $S_1 = \{1,2\}, S_2 = \{0,1\}, S_3 = \{1,2,3\}$.
The corresponding sets $A_x$ are: $A_0 = \{2\}, A_1 = \{1,2,3\}, A_2 = \{1,3\}, A_3 = \{3\}$.
Highlighted in red, an edge leading to a node $v$ with $\cup_{x\in T^{\enc}_v} A_x = A_2 \cup A_3 = \{1,3\}$, meaning that the tries for (the encoded) $S_1$ and $S_3$ will contain a copy of the same edge.
\textbf{Right.} Using the prefix-free encoding $\enc$ to encode $S_1,S_2,S_3$ by packing their codes into three tries (gray edges do not belong to the tries and are shown only for completeness). The tries for $S_1,S_2,S_3$ contain in total 12 edges, so $\trie(\enc(\langle S_1, S_2, S_3\rangle)) = 12$. 
In red: the two copies of the red edge on the left part of the figure, highlighting the equivalence of the two formulations of our trie-encoding problem. As a matter of fact, this is an optimal ordered code.} 
\label{fig:optimum ordered}
\end{figure}
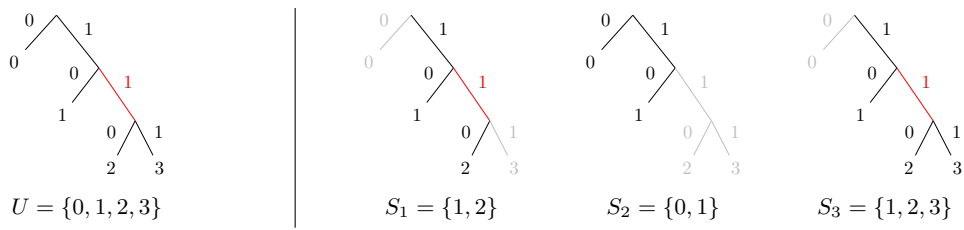

As there is a one-to-one relation between the set of all prefix-free binary encodings of the integers $U$ and the set of all binary trees with leaves $U$, we can conclude that the problem of finding a prefix-free encoding $\enc$ that minimizes $\trie(\enc(\mathcal S))$ can equivalently be seen as the problem of finding a binary tree $T$ with $u$ leaves (labeled with the universe elements $0,\dots, u-1$, not necessarily in this order) that minimizes the cost function $c(T) := \sum_{v\in T\setminus \{r(T)\}} | \bigcup_{x \in L(T_v)} A_x|$. 
We note that this problem is similar to the standard problems of (i) finding optimal binary search trees on a set of keys with given frequencies and (ii) finding the optimal prefix-free encoding for a source of symbols with given frequencies. As a matter of fact, if the sets $A_0, \ldots, A_{u - 1}$ are disjoint then our problem is equivalent to (i-ii) and Huffman's algorithm finds the optimal solution. 

Observe that this reformulation of the problem is equivalent, in turn, to the following optimization problem: given $u$ sets $A_0, \dots, A_{u-1}$, find the minimum-cost sequence of (binary) set unions that merges all sets into $\cup_{i=0}^{u-1} A_i$, where the cost of merging sets $A$ and $A'$ is $|A|+|A'|$. Ghosh et al.\ in \cite{Ghosh2015} proved this problem to be NP-hard, and provided tight approximations. This motivates us to study less general families of prefix-free encodings, which hopefully can be optimized more efficiently.

\section{Optimal Shifted Encoding}\label{section: encoding shifts}

Given a sequence $\mathcal{S}=\langle S_1,\cdots,S_n\rangle$ of $n$ subsets of $U$, in this section we study the problem of finding an optimal shift $a\in U$ that minimizes $\trie(\mathcal{S}+a)$. 
After finding a useful reformulation of $\trie(\mathcal{S}+a)$,
we describe an algorithm for finding an optimal shift $a$. 
The algorithm is parameterized on an abstract data structure for integer sequences. 
Using a simple array, we obtain an $O(u+N\log u)$-time algorithm. 
A DAG-compressed segment tree, on the other hand, gives an $O(N\log^2u)$-time algorithm.

\subsection{Trie measure as a function of the shift \texorpdfstring{$a$}{a}}
\label{subsection: trie in terms of shift}

Let $S=\{x_1,x_2,\dots,x_m\}\subseteq U$ be a non-empty integer set with $0\le x_1<x_2<\cdots<x_m < u$. We assume $\enc(x)$ is the standard $(\log u)$-bit binary representation of $x$ throughout this section.
Observe that for every non-negative integers $0\le x<y < u$, it holds that
\begin{equation}
    |\enc(x)\oplus \enc(y)|=\max_{j\in [x, y)}|\enc(j)\oplus\enc(j+1)|.
\label{eq: x oplus y as max of j oplus j+1}
\end{equation}

Consider the (infinite) sequence $\langle c_j\rangle_{j\ge 0}$ with $c_j=|\enc(j\modulo u)\oplus\enc((j+1)\modulo u)|$. 
This sequence is the infinite copy of the first $u/2$ elements  (see the first row in Figure~\ref{fig:organ}) of the sequence known as ``ruler function'' (OEIS sequence A001511\footnote{\url{https://oeis.org/A001511}}).
This sequence can be decomposed into the sum of $\log u$ periodic binary sequences. For integers $k\in[\log u]$ and $j\ge 0$, let us define $c^{(k)}_j$ as follows.
\begin{equation*}
    c^{(k)}_j = 
    \begin{cases}
        1 & \mbox{ if $j+1$ is a multiple of $2^{k-1}$,}\\
        0 & \mbox{ otherwise.}
    \end{cases}
\end{equation*}
Then, for every $j\ge 0$, it holds that $
c_j = \sum_{k=1}^{\log u} c^{(k)}_j$. 
Together with~\eqref{eq: x oplus y as max of j oplus j+1}, we obtain
\begin{equation}
    |\enc(x)\oplus \enc(y)|=\max_{j\in [x, y)}\sum_{k=1}^{\log u} c^{(k)}_j=\sum_{k=1}^{\log u} \max_{j\in [x, y)}c^{(k)}_j,
    \label{eq: trie x y = summax}
\end{equation}
\begin{figure}
    \centering
    \input{figures/organ-revised}
    \caption{The nodes in the trie of $S=\{x_1, x_2, x_3, x_4\} = \{2,4,10,13\}$ in universe $u = 16$. The $i$-th box on each row spans range $[x_{i}, x_{i+1})$, with $x_5 := x_1 + u$. The number of edges in the trie of the set is equal to the number of shaded boxes that contain at least one 1-bit. The shaded boxes correspond to the edges with the same color.
    }
    \label{fig:organ}
\end{figure}
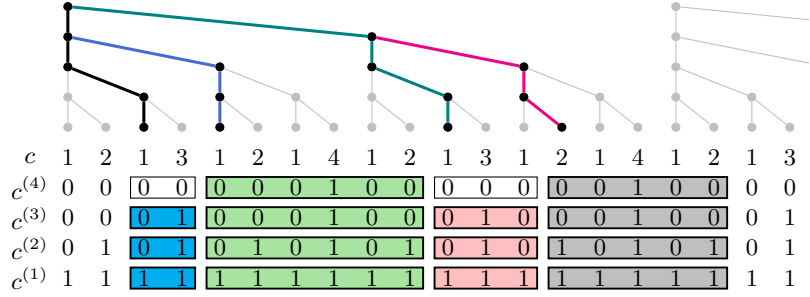%
where we used that the maximum and the sum are interchangeable since, for every $k>1$, if $c^{(k)}_j=1$ then $c^{(k-1)}_j=1$ by definition.

Recalling that the value of $c^{(k)}_j$ is either 0 or 1, this representation allows us to represent the cost (i.e.\ the number of edges) of a trie for $S=\{x_1,\cdots,x_m\}$ as follows. Consider the $\log u$ binary sequences $(c^{(k)}_{j})_{j\ge 0}$ for $k\in[\log u]$. For each sequence, and for every $1\le i\le m$, consider a range 
$[x_i,x_{i+1})$ of indices of the sequence $(c^{(k)}_j)_{j\ge 0}$ where we define $x_{m+1}=x_1+u$. 
For each $x_{i+1}$ with $i\in[m-1]$, observe that adding $x_{i+1}$ to the trie containing the integers $\{x_1,x_2,\cdots,x_{i}\}$ creates a new edge at level $k$ if and only if $\max_{j\in [x_i,x_{i+1})}c^{(k)}_j=1$.  Notice also that $x_1$ always creates $\log u$ edges forming the left-most path of the trie; at the same time, it always holds that $\max_{j\in [x_m,x_{m+1})}c^{(k)}_j=1$ for all $k\in [\log u]$ because $c^{(k)}_{u-1}=1$ for every $k\in[\log u]$. See Figure~\ref{fig:organ} for an example. 
More generally, the cost of a trie shifted by $a$ with $a\ge 0$ can be represented as in the following lemma whose proof is deferred to Appendix~\ref{app: deferred lemma trie cj proof}.

\begin{restatable}{lemma}{lemmatrieeqsumofc}
    \label{lemma: trie S+a in terms of c^k_j+a}
    Let $S=\{x_1,\cdots,x_m\}$ be a set of $m$ integers with $0\le x_1 <\cdots<x_m<u$. Let us define $x_{m+1}:=x_1+u$. For every $a\in U$, it holds that
    \[
        \trie(S+a)=\sum_{k=1}^{\log u}\sum_{i=1}^m \max_{j\in [x_{i} + a, x_{i+1} + a)} c^{(k)}_{j}
    \]
\end{restatable}

Lemma~\ref{lemma: trie S+a in terms of c^k_j+a} allows us to compute $\trie(S+a)$ by summing over $k\in[\log u]$ the number of the shifted ranges $[x_{i}+a,x_{i+1}+a)$ (for $i\in[m]$) such that $\max_{j\in[x_{i}+a,x_{i+1}+a)}c^{(k)}_j=1$; see Figure~\ref{fig: trie boxes} for an example. 
The following lemmata give an illustration of which values of $a$ involve a cost of 1 with a range $[x_{i},x_{i+1})$.

\begin{figure}
    \centering
    \resizebox{0.8\textwidth}{!}{\input{figures/costinterval-a.tex}}
    \caption{Representation of $\trie(S+a)$ based on $c^{(k)}_j$ for $S=\{2,4,10,13\}$ at level $k=3$ and $u=16$. Each row represents $c^{(3)}_{j+a}$,
    for $a\in U$. Boxes indicate ranges $[x_i,x_{i+1})$ and red boxes contain at least one 1. As an example, consider the pair $x_1, x_2$ (leftmost boxes). Among the shifts $0,\ldots,4$, the only shifts for which we do not pay the cost of $k=3$ for this pair are $a=2$ and $a=3$; i.e., $x_1+a$ and $x_2+a$ share an edge at level $3$ for $a=2,3$ in $\trie(S + a)$. Hence, in the figure we have two non-red boxes in the rows corresponding to $a=2,3$.}
    \label{fig: trie boxes}
\end{figure}

\begin{lemma} \label{lemma: always costs}
     Let $k\in [\log u]$, and $x,y\in [0,2u)$ with $x+2^{k-1}\le y$. Then it holds for every $a\in U$ that $\max_{j\in[x+a,y+a)}c^{(k)}_j=1$.
\end{lemma}
\begin{proof}
    Immediate from that any interval of length $2^{k-1}$ contains a multiple of $2^{k-1}$.
\end{proof}

\begin{lemma} \label{lemma: cost interval}
    For $a\in U$, $k\in [\log u]$, and $x,y\in [0,2u)$ with $x <y$, it holds that $\max_{j\in[x+a,y+a)}c^{(k)}_j=1$ if and only if $a \equiv_{2^{k-1}} b$ for some $b\in [2u - y, 2u - x)$. 
\end{lemma}
\begin{proof}
    Recall that $\max\{c^{(k)}_j:j\in[x+a,y+a)\}=1$ if and only if there exists $j \in [x+a, y+a)$ such that $j+1$ is a multiple of $2^{k-1}$ by definition of $c^{(k)}_j$.
    Assume that $j\in [x+a, y+a)$ is such that $j + 1$ is a multiple of $2^{k-1}$. Equivalently, $j=t\cdot 2^{k-1}-1$ for some integer $t$ and $x + a < t \cdot 2^{k-1} \le y + a$. The latter is equivalent to $a \in [t \cdot 2^{k-1} - y, t \cdot 2^{k-1} - x)$. This is then equivalent to $a\equiv_{2^{k-1}} b$ 
    for some $ b\in [- y, - x)$. Using that $2u$ is a multiple of $2^{k-1}$ and $k\in [\log u]$, this is in turn equivalent to $a \equiv_{2^{k-1}} b$ for some $b\in [2u - y, 2u - x)$.
\end{proof}

Considering $\max_{j\in[x_{i}+a,x_{i+1}+a)}c^{(k)}_j$ as a function of $a$, observe that it has a period of $2^{k-1}$ because of the periodicity of $c_j^{(k)}$. This means that we do not need to consider all the values of $a\in U$ but we can consider only $a\in[0,2^{k-1})$.
For each pair of consecutive elements $x_i$ and $x_{i+1}$ for $i\in[m]$, 
let us define $I_k(x_i,x_{i+1})$ as:
\begin{equation}
    I_k(x_i,x_{i+1})=
    \begin{cases}
        [0,2^{k-1}) & \mbox{ if } x_{i+1}-x_{i} \ge 2^{k-1},\\
        [\ell,r) & \mbox{ if } \ell<r, \\ 
        [0,r) \cup [\ell,2^{k-1}) & \mbox{ otherwise.}
    \end{cases}
    \label{eq: I_k}
\end{equation}
where $\ell:=(2u-x_{i+1})\modulo 2^{k-1}$, and 
$r:=(2u-x_{i})\modulo 2^{k-1}$. 
Note that $I_k(x_{i},x_{i+1})$ can be represented with one or two intervals. 

Let $\mathcal{I}_k$ be the multiset $\mathcal{I}_k = \lBrace I_k(x_i,x_{i+1})\ :\ i\in [m] \rBrace$. 
Consider the number of edges of the trie at level $k$ for the shifted set $S+a$ for a specific point $a\in U$.
By the inner sum of Lemma~\ref{lemma: trie S+a in terms of c^k_j+a} along with Lemma~\ref{lemma: always costs} and Lemma~\ref{lemma: cost interval}, this number is equal to the number of members of $\mathcal{I}_k$ containing a specific point $a\modulo 2^{k-1}$, i.e., to the \emph{depth} of $\mathcal{I}_k$ at position $a\modulo 2^{k-1}$. 
In the following lemma, we show that $\trie(S+a)$ can be expressed as the sum over $k\in[\log u]$ of the depth of $\mathcal{I}_k$ at position $a\modulo 2^{k-1}$.

\begin{lemma} \label{lemma: trie periodic interval depth}
    Let $S=\{x_1,x_2,\cdots,x_m\}$ be a set of integers with $0\le x_1<x_2<\cdots<x_m<u$. Let $x_{m+1}:=x_1+u$ and $a\in U$. Then,
    \[
        \trie(S+a)=\sum_{k=1}^{\log u} \sum_{i=1}^m \ones[(a \modulo 2^{k-1}) \in I_k(x_i, x_{i + 1})].
    \]
\end{lemma}

\begin{proof} 
    According to Lemma~\ref{lemma: trie S+a in terms of c^k_j+a}, we have $\trie(S+a)=\sum_{k=1}^{\log u}\sum_{i=1}^m \max_{j \in [x_{i}+a, x_{i+1}+a)} c^{(k)}_{j}$. Hence it is sufficient to show that $a\modulo 2^{k-1}\in I_k(x_i,x_{i+1})$ iff $\max_{j\in[x_i+a,x_{i+1}+a)} c^{(k)}_{j}=1$. We distinguish two cases: (Case 1) If $x_{i+1}-x_{i}\ge 2^{k-1}$, we have $ I_k(x_{i},x_{i+1})=[0,2^{k-1})$ and thus clearly $a\modulo 2^{k-1} \in I_k(x_{i},x_{i+1})$. At the same time $\max_{j\in[x_i + a,x_{i+1} + a)} c^{(k)}_{j}=1$ by Lemma~\ref{lemma: always costs} for any $a\in U$.
    (Case 2) Assume $x_{i+1}-x_{i}<2^{k-1}$. Let $\ell = (2u-x_{i+1}) \modulo 2^{k-1}$ and $r=(2u-x_{i})\modulo 2^{k-1}$ as in Eq.~\eqref{eq: I_k}.  
    By Lemma~\ref{lemma: cost interval}, it holds that $\max_{j\in[x_i + a,x_{i+1} + a)} c^{(k)}_{j}=1$ if and only if $a \equiv_{2^{k-1}} b$ for some $b \in [2u -x_{i+1}, 2u - x_i)$. 
    The latter is equivalent to: 
    \begin{equation}
        a \equiv_{2^{k-1}} b\mbox{ for some }b \in [\ell,\ell+x_{i+1}-x_{i})\label{eq:a eq b}
    \end{equation}
    We have two cases.
    (Case 2a) Assume that $\ell<r$ holds. Since $x_{i+1}-x_{i}<2^{k-1}$, $\ell<r$ if and only if $\ell+x_{i+1}-x_{i}<2^{k-1}$. Observing that $r=(\ell+x_{i+1}-x_{i})\modulo 2^{k-1}$, we have $\ell+x_{i+1}-x_{i}=r$.
    (Case 2b) Now assume that $\ell \ge r$.
    Then Eq.~\eqref{eq:a eq b} is equivalent to $a \equiv_{2^{k-1}} b$ with $b \in [\ell, 2^{k -1})$ or $b \in [2^{k-1}, r+2^{k-1})$, which can be rewritten as $a \equiv_{2^{k-1}} b$ with $b \in [0, r)$ or $b \in [\ell, 2^{k -1})$. This completes the proof.
\end{proof}
Lemma~\ref{lemma: trie periodic interval depth} can be naturally generalized to a sequence of sets $\mathcal{S}+a=\langle S_1+a,\cdots,S_n+a\rangle$. Together with Definition~\ref{def:trie prefix-free sequence}, we obtain:
\begin{equation}
    \trie(\mathcal{S}+a)=\sum_{i=1}^{n}\trie(S_i+a)
    =\sum_{k=1}^{\log u} \sum_{i=1}^n \sum_{j=1}^{|S_i|} \ones[(a \modulo 2^{k-1}) \in I_k(x^{(i)}_j,x^{(i)}_{j+1})].
    \label{eq: trie S+a generalization} 
\end{equation}
where $x^{(i)}_j$ is the $j$-th smallest element of $S_i$.
In the following subsection, we develop algorithms to find an optimal shift $a\in[0,u)$ minimizing $\trie(\mathcal S+a)$ using this formulation.

\subsection{Algorithms for the optimal shift}
\label{subsection: optimum shift}

Let $S_1,\cdots,S_n\subseteq U$ be $n$ sets of integers. For $i\in [n]$ and $j\in [|S_i|]$, let $x^{(i)}_j$ denote the $j$-th smallest element of $S_i$. For $k\in [\log u]$, let $D_k[0..2^{k-1})$ be a sequence of length $2^{k-1}$ such that, for $a \in[0,2^{k-1})$,
\begin{equation}
    D_k[a] = 
    \sum_{i=1}^n \sum_{j=1}^{|S_i|} \ones[(a \modulo 2^{k-1}) \in I_k(x^{(i)}_j,x^{(i)}_{j+1})].
    \label{eq: D_k interval seg}
\end{equation}
Observe that $D_k[a]$ is defined as the two inner sums of Eq.~\eqref{eq: trie S+a generalization}; in other words, $D_k[a]$ is the number of interval segments at level $k$ that contain a specific position $a\in[0,2^{k-1})$. To consider the cumulative sum of the number of interval segments that contain position $a$ up to level $k\in[\log u]$, let $C_k$ be a sequence of length $2^k$ that, for $a\in[0,2^k)$, is defined as
\begin{equation}
    C_k[a]=\sum_{k'=1}^{k} D_{k'}[a \modulo 2^{k'-1}].
    \label{eq: Ck is sum of Dk}
\end{equation}
Then, by Equations \eqref{eq: trie S+a generalization} and \eqref{eq: D_k interval seg} it holds that $\trie(\mathcal{S}+a) = C_{\log u}[a]$ for every $a\in U$.
Therefore, if we can compute $C_k$ in an efficient way, this will allow us to find an optimal shift by computing $\arg\min_{a\in U}C_{\log u}[a]$. To do this, now consider an abstract data type $\mathcal{D}$ that supports the following four operations:
\begin{enumerate}
    \item $\mathcal{D}.\mathtt{initialize}()$: create a sequence $A$ of length 1 and initialize it as $A[0]\gets 0$.
    \item $\mathcal{D}.\mathtt{add}(\ell,r)$: update $A[a]\gets A[a]+1$ for $a\in [\ell,r)$.
\item $\mathcal{D}.\mathtt{extend}()$: duplicate the sequence as $A\gets AA$.
    \item $\mathcal{D}.\mathtt{argmin}()$: return $\arg\min_{0\le a<|A|}A[a]$.
\end{enumerate}
\begin{algorithm}[t]
$\mathcal{D}.\mathtt{initialize}()$ \tcp*{Initialize and start with universe of size $2^0$}
\For{$k = 1 .. \log u$}{
    \For{$i = 1 .. n$}{
        \For{$j = 1 .. |S_i|$}{
            \For{each interval $[\ell,r)$ forming $I_k(x^{(i)}_j,x^{(i)}_{j+1})$ \label{alg: line: foreach}}{
                $\mathcal{D}.\mathtt{add}
                (\ell,r)$   \tcp*{Add interval segments} \label{alg: line: D add}
            }
        }
    }
    $\mathcal{D}.\mathtt{extend}()$ 
    \tcp*{Extend the universe size into $[0,2^k)$} \label{alg: line: D extend}
}
\textbf{return} $\mathcal{D}.\mathtt{argmin}()$ \label{alg: line: D argmin}\;
\caption{General algorithm for finding an optimal shift of $S_1,\cdots,S_n$.} \label{alg: general procedure}
\end{algorithm}
Based on these operations, in
Algorithm~\ref{alg: general procedure} we describe a general procedure to find an optimal shift. For each level $k\in [\log u]$ and every set $S_i$, we iterate on every pair of consecutive elements $x^{(i)}_j, x^{(i)}_{j+1}$. In Line~\ref{alg: line: foreach}, we use the fact that $I_k(\cdot)$ can be represented as the union of at most two intervals on $U$, see Eq.~\eqref{eq: I_k}. We increment $A[a]$ by 1
for every $a\in[\ell,r)$ by calling $\mathcal{D}.\mathtt{add}(\ell,r)$ 
(Line~\ref{alg: line: D add}). 
After processing all consecutive pairs, we extend the universe into $[0,2^k)$ (Line~\ref{alg: line: D extend}). To see that the algorithm is correct, observe that the amount by which $A[a]$ is incremented in the three inner loops equals $D_k[a]$ 
according to Eq.~\eqref{eq: D_k interval seg}. Eq.~\eqref{eq: Ck is sum of Dk} yields that, for $k\in[\log u]$ and $a\in[0,2^{k-1})$, the number $C_k$ can be written recursively as 
\begin{equation}
    C_{k}[a+2^{k-1}]=C_{k}[a]=C_{k-1}[a] + D_k[a],
    \label{eq: C_k recursive}
\end{equation}
where we define $C_0:=\langle 0\rangle$ as the base case.
Then, at the end of the $k$-th iteration of the outer for loop (after Line~\ref{alg: line: D extend}), the sequence $A$ 
represented by $\mathcal{D}$ is exactly $C_k$. The running time depends on how $\mathcal{D}$ is implemented. In the following two subsections, we will present two different ways of implementing the data structure $\mathcal{D}$. These allow us to find an optimal shift in $O(u+N\log u)$ with simple arrays and $O(N\log^2u)$ time with a dynamic DAG structure.

\subsubsection{\texorpdfstring{$O(u+N\log u)$}{O(u+N log u)}-time algorithm}
\label{subsubsection: u+Nlogu}

Our first solution is simple: we represent $A$ implicitly by storing an array $\Delta[0..|A|-1]$ 
of length $|A|$, encoding the differences between adjacent elements of $A$.
At the beginning, $\Delta$ is initialized as an array of length 1 containing the integer 0. 
Operation $\mathcal{D}.\mathtt{add}(\ell,r)$ is implementing by incrementing $\Delta[\ell]$ by 1 unit and (if $r<|\Delta|$) decrementing $\Delta[r]$ by 1 unit, in $O(1)$ time (Lines~\ref{alg: u+Nlogu line: add1}--\ref{alg: u+Nlogu line: add2} of Algorithm \ref{alg: u+Nlogu} in Appendix \ref{app: deferred u+Nlogu}). 
As far as operations $\mathcal{D}.\mathtt{extend}()$ and $\mathcal{D}.\mathtt{argmin}()$ are concerned, they can be supported naively in linear $O(|\Delta|) = O(|A|)$ time with a constant number of scans of array $\Delta$. See Algorithm \ref{alg: u+Nlogu} for the details.

Next, we analyze the running time of Algorithm~\ref{alg: general procedure} when using this simple implementation of $\mathcal D$. Operation $\mathcal{D}.\mathtt{add}()$ is called $O(N\log u)$ times, each of which costs $O(1)$ time. 
Operation $\mathcal{D}.\mathtt{extend}()$ is called $\log u$ times, and each call costs $O(|A|)$ time. Each time this operation is called, the length of $A$ doubles (starting from $|A|=1$). The overall cost of these calls is therefore $O(1+2+4+\dots + u) = O(u)$. Finally, 
$\mathcal{D}.\mathtt{argmin}()$ is called only once when $|A| = u$, and  therefore it costs $O(u)$ time. 
We conclude that, using this implementation of $\mathcal D$, Algorithm \ref{alg: general procedure} runs in $O(u + N\log u)$ time. Within this running time we actually obtain a much more general result, since we can evaluate any entry of $A$:
\begin{theorem}
    Given a sequence $\mathcal{S}=\langle S_1,\cdots,S_n\rangle$ of $n$ subsets of $U$, we can compute $\trie(\mathcal{S}+a)$ for all $a\in U$ in total $O(u + N\log u)$ time.
\end{theorem}

\subsubsection{\texorpdfstring{$O(N\log^2 u)$}{O(N log2 u)}-time algorithm}
\label{subsubsection: Nlog2u}

If $u$ is too large compared to $N$, then the simple solution presented in the previous subsection is not efficient. 
In this case, we instead represent $A$ with a DAG-compressed variant of the segment tree \cite{bentley80,kline95}; 
see Algorithm \ref{alg: Nlog2u 2} in Appendix~\ref{app: deferred Nlog2u} for the details.
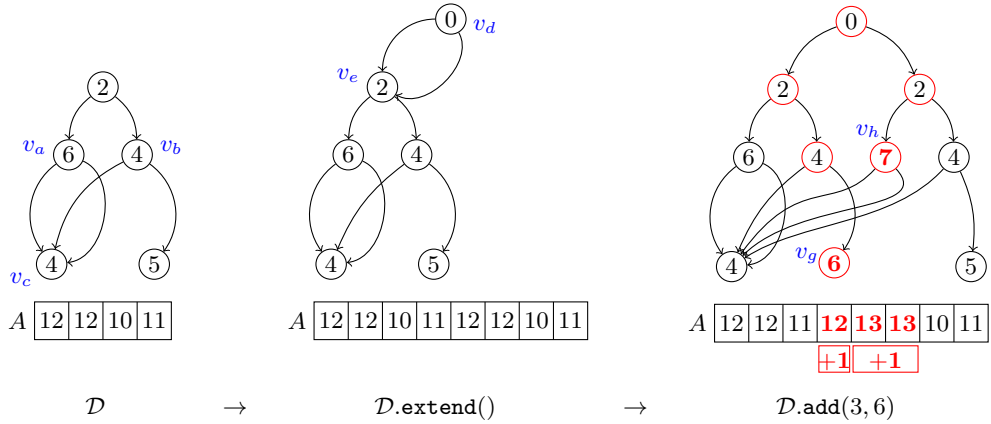
\begin{figure}[t]
    \begin{center}
    \setlength\extrarowheight{-10pt}
    \begin{tabular}{ccccc}
    \input{figures/grammartree1.tex}&&
    \input{figures/grammartree2.tex}&&
    \input{figures/grammartree3-36.tex}\\
    $\mathcal{D}$
    &$\rightarrow$
    &$\mathcal{D}.\mathtt{extend}()$
    &$\rightarrow$
    &$\mathcal{D}.\mathtt{add}(3,6)$\\
    \end{tabular}
    \end{center}
    \caption{DAG-compressed representation of $A$. (Left) Each element $A[i]$ is represented as the sum of the values stored in its corresponding root-to-leaf path. The two children of $v_a$ and the left child of $v_b$ are represented by a single node $v_c$. (Middle) Duplicating the whole content can be performed by creating a new root ($v_d$) with the old root ($v_e$) being referred as both of its left and right children. (Right) Incrementing $A[i]$ by 1 for each $i\in[3,6)=[3,5]$ is performed by incrementing the values in $v_g$ and $v_h$ by 1 each, which covers $[3,3]$ and $[4,5]$, respectively. We duplicate every node that has more than one incoming edges when it is visited so that the visited nodes (indicated with red nodes) should have unique paths from the root.}
\end{figure}
Intuitively, consider the complete binary tree of height $\log|A|$ where the root stores a counter associated to the whole array $A$, its left/right children store a counter associated to $A[0,|A|/2-1]$ and $A[|A|/2, |A|-1]$, respectively, and so on recursively (up to the leaves, which cover individual values of the array). 
Assume that these counters are initialized to 0.
Observe that for every $0\le \ell<r \le |A|$, there exists the coarsest partition of $A[\ell,r-1]$ into $O(\log |A|)$ disjoint intervals covered by nodes of the tree.
The idea is to support $\mathcal{D}.\mathtt{add}(\ell,r)$ by incrementing by 1 unit the counter associated to those nodes. 
To avoid spending time $O(|A|)$ for operation  $\mathcal{D}.\mathtt{extend}()$, however, we cannot afford duplicating the whole tree when this operation is called. In this case, our idea is simple: since $\mathcal{D}.\mathtt{extend}()$ duplicates the whole content of $A$, in $O(1)$ time we simply create a new node covering $AA$ and make both its two outgoing edges (left/right child) lazily point to the node covering $A$, i.e., the old root of the tree. In other words, we represented the tree as a DAG, collapsing nodes covering identical sub-arrays. This modification makes it necessary to (possibly) duplicate at most $O(\log |A|)$ nodes at each call of $\mathcal{D}.\mathtt{add}(\ell,r)$: a duplication happens when, starting from the root, we reach a node $x$ having more than one incoming edges. Without loss of generality, assume that we are moving from $y$ that covers sub-array $A[i,i+2^{k+1}-1]$ to its left child. Since $x$ has more than one incoming edges, $A[i,i+2^{k}-1]$ is not the unique interval that $x$ covers. We create a new node $x'$ by duplicating $x$, then make $x'$ the left child of $y$ so that $A[i,i+2^{k}-1]$ is the unique interval that $x'$ covers. Then proceed to $x'$. 
Since $A[\ell,r-1]$ is covered by $O(\log|A|)$ nodes, starting this procedure at the root duplicates at most $O(\log|A|)$ nodes of the DAG and operation $\mathcal{D}.\mathtt{add}(\ell,r)$ therefore costs $O(\log|A|) \subseteq O(\log u)$ time. This slow-down (with respect to the $O(1)$ cost of the same operation in the previous subsection) is paid off by the fact that we do not need to create $O(|A|)$ new nodes (i.e.\ duplicate the tree) at each call of $\mathcal{D}.\mathtt{extend}()$. 
With this representation, the operation $\mathcal{D}.\mathtt{argmin}()$ can be implemented in $O(1)$ time as well by simply associating to each node the smallest value in the sub-array of $A$ covered by that node, as well as its index (these values are updated inside $\mathcal{D}.\mathtt{add}(\ell,r)$). Operations $\mathcal{D}.\mathtt{initialize}()$ and $\mathcal{D}.\mathtt{extend}()$ trivially take $O(1)$ time.

Plugging this structure into Algorithm~\ref{alg: general procedure}, observe that the running time is dominated by $\mathcal{D}.\mathtt{add}(\cdot)$, which is called $O(N\log u)$ times. Since each call to this function takes $O(\log u)$ time as discussed above, we finally obtain: 

\begin{theorem}
    Given a sequence $\mathcal{S}=\langle S_1,\cdots,S_n\rangle$ of $n$ subsets of $U$, we can compute one value $a$ minimizing $\trie(\mathcal{S}+a)$ in $O(N\log^2 u)$ time.
\end{theorem}

Finally, we remark that using more elaborate arguments one can obtain an algorithm running in $O(N\log u)$ time. Due to space limitations, we refrain from detailing this approach here, but will include it in the extended version of the article.

\section{Optimal Ordered Encoding}\label{sec:optimum ordered}

To compute the optimal \emph{ordered encoding}, we employ the equivalent problem formulation of Section \ref{sec:equivalent problem formulation}.
For any $x\in U$, let $A_x = \{i\in [n]\ :\ x\in S_i\}$. Recall that we can assume, w.l.o.g., that $A_x\neq \emptyset$ for all $x\in U$ and $S_i\neq \emptyset$ for all $i\in [n]$. 
Our goal is to find the binary tree $T$ with leaves $0, \ldots, u - 1$ (in this order) minimizing $c(T) := \sum_{v\in T\setminus \{r(T)\}} | \bigcup_{x \in L(T_v)} A_x|$.
For a binary tree $T$, let $T_\ell$ and $T_r$ be the left and right sub-tree of the root of $T$, respectively.
We show that the dynamic programming solution of Knuth \cite{knuth1971optimum} can be applied to our scenario. 
Consider the following alternative cost function: 
$d(T) := \Big|\bigcup_{x \in L(T)} A_x\Big| + d(T_\ell) + d(T_r)$.
The (recursive) function $d(T)$ is related to $c(T)$ by the following equality: $c(T) = d(T) - |\bigcup_{x \in L(T)} A_x|$.
It is not hard to turn the definition of $d(T)$ into a set of dynamic programming formulas computing both the best tree $T$ and its cost $d(T)$ by creating, for every $0 \le x\le y < u$, a variable $d_{x,y}$ whose final value will be $\min_{T'} d(T')$ (where $T'$ runs over all trees such that $L(T') = \{x,x+1, \dots, y\}$), and a pre-computed constant $a_{x,y} = \Big|\bigcup_{t=x}^{y} A_t \Big|$.
Given the constants $a_{x,y}$, the following set of dynamic programming formulas finds (bottom up, i.e., by increasing $y-x$) the value $c(T) = d(T) - |\bigcup_{x \in L(T)} A_x| = d_{0, u-1} - a_{0,u-1}$ for the optimal tree $T$ and, by standard backtracking, the tree $T$ itself:
\begin{align*}
	\text{initialization:}\qquad 
	a_{x,y}	&:=\quad \Big|\bigcup_{t=x}^{y} A_t \Big|	&\text{for all }0\le x \le y \le u - 1\\
	d_{x,x}	&:=\quad a_{x,x}						&\text{for all }x \in U\\
	\text{recursion:}\qquad
	d_{x,y}	&:=\quad a_{x,y} + \min_{x<z\le y} d_{x, z - 1} + d_{z, y}					&\text{for all }0\le x < y \le u - 1
\end{align*}

The optimal tree's topology can be retrieved by backtracking: if $L(T') = \{x,x+1, \dots, y\}$, then the number of leaves in $T'_\ell$ 
is equal to 
$(\argmin_{x<z\le y} d_{x, z - 1} + d_{z, y})-i$ (we start from the root of $T$ with $x=0$ and $y=u-1$).
Below, we show how to compute constants $a_{x,y}$ in $O(N+u^2)$ time. Since the above formulas can be evaluated bottom-up (i.e.\ by increasing $y-x$) in $O(u^3)$ time\footnote{\scriptsize In his work \cite{knuth1971optimum}, Knuth shows how to evalutate these recursive formulas to $O(u^2)$ time. In the simpler case considered in his article (corresponding to our problem with pairwise-disjoint $A_0, \dots, A_{u-1}$), one can bound $\argmin_{x<z\le y} d_{x, z - 1} + d_{z, y}$ so that the overall cost of looking for all those optimal values of $z$ amortizes to $u^2$. Unfortunately, in our scenario we haven't been able to prove that the same technique can still be applied.}, the overall running time of the algorithm is $O(N + u^3)$.

We now show how to compute efficiently the constants $a_{x,y}$.
We describe an overview of the algorithm
(see Algorithm \ref{alg:unions} for the pseudocode).
We add to the sequence two sets $A_{-1} = A_u = [n]$, respectively at the beginning and end: the new sequence becomes $A_{-1}, A_0, \dots, A_{u-1}, A_u$. 
The main idea behind the algorithm is to initially compute inside $a_{x,y}$ the number of elements from $[n]$ that are \emph{missing} from $\Big|\bigcup_{t=i}^{j} A_t \Big|$. 
Then, the substitution $a_{x,y} \gets n - a_{x,y}$ will yield the final result. 

\begin{algorithm}[t]
	\caption{Compute $a_{x,y}$ }
        \label{alg:unions}
	\SetKwInOut{Input}{input}
	\SetKwInOut{Output}{output}
	\SetSideCommentLeft
	\LinesNumbered

    \Input{Sets $A_0, \dots, A_{u-1} \subseteq [n]$ of total cardinality $N = \sum_{x=0}^{u-1} |A_x|$, s.t. for every $i\in[n]$ there exists $x\in U$ with $i\in A_x$.}
    \Output{$a_{x,y} = \Big|\bigcup_{t=x}^{y} A_t \Big|$, for each $0\le x \le y < u$}

    \BlankLine

    Add dummy sets $A_{-1} = A_u = [n]$\;

    \textbf{for each } $0\le x,y\le u$ \textbf{do} $a_{x,y} \leftarrow 0$\;
 
    $P[1,\dots,n] \leftarrow -1$\tcp*[r]{Previous occurrence of $i\in[n]$, all initialized to $-1$}

    \BlankLine

    \For{$y = 0, \dots, u$}{

        \For{$i\in A_y$}{

            $x \leftarrow P[i]$ \tcp*[r]{$A_{x+1}, \dots, A_{y-1}$: maximal sequence not containing $i$}
            $P[i] \leftarrow y$\;

            \If{$x<y-1$}{

                $a_{y-1,y-1} \leftarrow a_{y-1,y-1} + 1$\;
                $a_{x,x} \leftarrow a_{x,x} + 1$\;
                $a_{x,y-1} \leftarrow a_{x,y-1} - 1$\; 
                $a_{y-1,x} \leftarrow a_{y-1,x} - 1$\;
            
            }

        }
    
    }

    \BlankLine

    \For{$x=0,\dots,u-1$}{
        \For{$y=u-2, \dots, 0$}{
            $a_{x,y} \leftarrow a_{x,y} +  a_{x,y+1}$\tcp*[r]{Partial sums, row-wise}
        }
    }

    \BlankLine

    \For{$y=0,\dots,u-1$}{
        \For{$x=u-2, \dots, 0$}{
            $a_{x,y} \leftarrow a_{x,y} +  a_{x+1,y}$\tcp*[r]{Partial sums, column-wise}
        }
    }

    \textbf{for each } $0\le x,y< u$ \textbf{do} $a_{x,y} \leftarrow n-a_{x,y}$\;

    \Return $a_{x,y}$ for all $0\le x \le y < u$\;
	
\end{algorithm}

We initialize $a_{x,y} \gets
0$ for all $0\le x,y \le u$. 
Let $A_{x+1}, A_{x+2}, \dots, A_{y-1}$ be a maximal contiguous subsequence of sets not containing a given $i \in [n]$, that is, (i) $i\notin A_{x'}$ for all $x'=x+1, x+2, \dots, y-1$, (ii) $i\in A_{x}$,  and (iii) $i \in A_{y}$. Then, $i \notin \quad \Big|\bigcup_{t=x'}^{y'} A_t \Big|$ for every $i\le x' \le y' \le y$. 
Imagine then adding one unit to $a_{x',y'}$ for all $(x',y') \in [x+1,y-1] \times [x+1,y-1]$. Clearly, if we can achieve this for every $i \in[n]$ and every such maximal contiguous subsequence of sets not containing $i$, at the end each $a_{x,y}$ will contain precisely the value $n - \Big|\bigcup_{t=x}^{y} A_t \Big|$.
The issue is, of course, that adding one unit to $a_{x',y'}$ for all $x < x' \le y' < y$, costs time $O(xy)$ if done naively. The crucial observation is that this task can actually be performed in constant time using (bidimensional) \emph{partial sums}: see Figure~\ref{fig:partial sums} for an example.  Ultimately, this trick allows us computing all $a_{x,y}$ in the claimed $O(N+u^2)$ running time. 

\begin{figure}[h!]
\scriptsize
\begin{minipage}{.32\linewidth}
\[\arraycolsep=3pt\def\arraystretch{1.2}
\begin{array}{c|rrrr}
\tikz{\node[below left, inner sep=1pt] (def) {$x$};%
      \node[above right,inner sep=1pt] (abc) {$y$};%
      \draw (def.north west|-abc.north west) -- (def.south east-|abc.south east);}
  & 0  &  1  & 2   & 3\\\hline
0 & 1  &  0  & -1   &  0 \\
1 & 0  &  1  & 0   & -1 \\
2 & -1 &  0  & 1   &  0 \\
3 & 0  & -1  & 0   &  1\\
\end{array}\]
\end{minipage}
\begin{minipage}{.32\linewidth}
\[
\arraycolsep=4pt\def\arraystretch{1.2}
\begin{array}{c|rrrr}
\tikz{\node[below left, inner sep=1pt] (def) {$x$};%
      \node[above right,inner sep=1pt] (abc) {$y$};%
      \draw (def.north west|-abc.north west) -- (def.south east-|abc.south east);}
  & 0 & 1 &  2  & 3\\\hline
0 & 0 & -1 &  -1  &  0 \\
1 & 0 & 0 & -1  & -1 \\
2 & 0 & 1 &  1  &  0\\
3 & 0 & 0 &  1  &  1\\
\end{array}\]
\end{minipage}
\begin{minipage}{.32\linewidth}
\[
\arraycolsep=6pt\def\arraystretch{1.2}
\begin{array}{c|rrrr}
\tikz{\node[below left, inner sep=1pt] (def) {$x$};%
      \node[above right,inner sep=1pt] (abc) {$y$};%
      \draw (def.north west|-abc.north west) -- (def.south east-|abc.south east);}
  & 0 & 1 & 2  & 3\\\hline
0 & 0 & 0 & 0  & 0 \\
1 & 0 & 1 & 1  & 0 \\
2 & 0 & 1 & 2 & 1\\
3 & 0 & 0 & 1 & 1\\
\end{array}\]
\end{minipage}\caption{
\footnotesize
Matrix $a_{x,y}$. Suppose our goal is to add 1 unit to the two sub-matrices with indices $[1,2]\times [1,2]$ and $[2,3]\times [2,3]$. We show how to achieve this by performing 4 updates for each of these two submatrices, and then performing two scans (partial sums) of total cost $O(u^2)$ over the full matrix. Letting $t$ be the number of sub-matrices to update (in this example, $t=2$), this means that we can perform the $t$ updates in total $O(t + u^2)$ time. 
Assume we start by a matrix $a_{x,y}$ containing only zeros.
To update the sub-matrix with indices $[x',y']\times [x',y']$, we perform these 4 operations: (1) $a_{y',y'} \leftarrow a_{y',y'}+1$, (2) $a_{x'-1,x'-1} \leftarrow a_{x'-1,y'-1} + 1$, (3) $a_{x'-1,y'} \leftarrow a_{x'-1,y'} - 1$, and (4) $a_{y',x'-1} \leftarrow a_{y',x'-1} - 1$. 
\textbf{Left:} applying these four operations for each of the two sub-matrices 
$[1,2]\times [1,2]$ and $[2,3]\times [2,3]$. \textbf{Center}: we compute partial sums row-wise, cumulating from right to left. \textbf{Right}: we compute partial sums column-wise, cumulating bottom-up. At the end, we correctly added 1 unit to the target sub-matrices.}\label{fig:partial sums}
\end{figure}

\subsection{Optimal Shifted Ordered Encoding}
\label{sec:optimum-shifted-ordered-encoding}
Observe that the optimal ordered encoding is not necessarily better than the optimal shifted encoding, because shifted encodings are not ordered (except the case $a=0$). It is actually very easy to get the best of both worlds and compute the ordered encoding $\enc : \{0,1\}^* \rightarrow \{0,1\}^*$ minimizing $\min_{a\in U} \trie(\enc(\mathcal S+a))$. This encoding is guaranteed to be no worse than \emph{both} the best shifted encoding and the best ordered encoding. 

The solution is a straightforward extension of the technique used for the best ordered encoding. 
Let $A_0, \dots, A_{u-1}$ be the sets defined previously. 
Build the sequence of sets $A_0, \dots, A_{u-1}, A_u, \dots, A_{2u-1}$, where $A_{u+x} = A_i$ for all $0\leq x < u$ (that is, we simply create an extra copy of the sets). Compute $a_{x,y}$ and $d_{x,y}$ as discussed previously, with the only difference that now the domain of $x,y$ is doubled: $0\leq x,y < 2u$. It is not hard to see that the optimal shifted ordered encoding has then cost $c'(T) = \min_{0\leq a < u} (d_{a,a+u-1} + a_{a,a+u-1}) = (\min_{0\leq a < u} d_{a,a+u-1}) + a_{0,u-1}$. Again, the tree topology is obtained by backtracking starting from the optimal interval $[a',a'+u)$ given by $a' = \argmin_{0\leq a < u} d_{a,a+u-1}$ (the optimal ordered encoding discussed previously is simply the particular case $a'=0$). The asymptotic cost of computing this optimal shifted ordered encoding is still $O(N+u^3)$.

\section{Experimental Results}
\label{sec:experiments}
We implemented our algorithms for computing the optimal shifted encoding and the optimal shifted ordered encoding in {\tt C++} and made them available at \url{https://github.com/regindex/trie-measure}. We computed the sizes of these encodings on thirteen datasets, the details of which can be found in Table \ref{tab:datasets} in Appendix \ref{sec:experiments}. 
The datasets can be naturally split in four groups according to their origin: 1) eight sequences of sets from \cite{Benson18}, 2) a dataset of paper tags from dblp.org xml dump \cite{dblpdump}, 3) a dataset containing amino acid sets from two protein sequence collections \cite{proteins-long,proteins-high-id}, 4) and two datasets containing ratings and tags of 10000 popular books \cite{goodbooks-10k}. For the last three groups, we generated sequences of sets (i.e. the inputs of our algorithms) by extracting sets of features from the original datasets: in 2) we extracted sets of tags for all dblp entries, in 3) we extracted the set of amino acids contained in each protein sequence, and in the two datasets of group 4) we extracted a set of tags and ratings for each book, respectively. The repository above contains all the generated set sequences.

As far as the shifted trie measure $\trie(\mathcal S + a)$ of Section \ref{section: encoding shifts} is concerned, we evaluated it on the above datasets for all shifts $a\in U$ and reported the following statistics in Table~\ref{tab:summary}: the trie cost for the optimal and worst shifts, the average cost over all shifts, and the percentage differences between the average/worst shifts and the optimal shift. 
Our results indicate that the optimal, average, and worst shifts lead in practice to similar costs. In particular, only five datasets showed a difference larger than 5\%  between the optimal and worst shift (opt-shift/worst-shift(\%) < 95). Among these, the differences range between 15.85\% for {\tt DBLP.xml} and 6.01\% for {\tt tags-math-sx-seqs}. The differences are even smaller when comparing the optimum with the average shift. In this case, only {\tt DBLP.xml} shows an average difference (of 6.95\%) being larger than 5\% (opt-shift/avg-shift(\%) < 95), meaning that the trie measure computed with a random shift on this dataset is, on average, 6.95\% larger than with the optimal shift. 
These results suggest that data structures which encode integer sets as tries using the standard binary integer encoding (such as the subset wavelet tree of Alanko et al. \cite{AlankoBPV23}) 
are often efficient in practice since even arbitrary shifts allow to obtain an encoding being not far from the optimal shifted encoding.

As far as the optimal shifted ordered encoding of Section~\ref{sec:optimum-shifted-ordered-encoding} is concerned, 
the last two columns of Table~\ref{tab:summary} show the size of this encoding for the ten datasets on which we could run our cubic dynamic programming algorithm. 
The results indicate that the size of this encoding tends to be significantly smaller than the optimal shifted encoding discussed in the previous paragraph: seven datasets showed a percentage difference between the optimal shifted encoding and the optimal shifted ordered encoding being larger than 10\%, with a peak of 27.31\% on {\tt tags-math-sx-seqs}. 

We leave it as an open question whether it is possible to compute the optimal shifted ordered encoding in sub-cubic time as a function of $u$. Another interesting research direction is to design fast heuristic (e.g. ILP formulations) for computing the globally-optimal prefix free encoding (NP-hard to compute).

\bibliography{article}

\newpage

\appendix

\section{Additional Material for Section~\ref{section: encoding shifts}}
\subsection{Proof of Lemma~\ref{lemma: trie S+a in terms of c^k_j+a}}
\label{app: deferred lemma trie cj proof}
\lemmatrieeqsumofc*
\begin{proof}
Let $a\in U$ be arbitrary.
We divide $S$ into $S_{<}=\{x\in S:x+a<u\}$ and $S_{\ge}=\{x-u:x\in S\mbox{ and }x+a\ge u\}$ to consider the effect of computing modulo $u$. We distinguish three cases.

First, consider the case where $S_{\ge}=\emptyset$. 
Observe that $c^{(k)}_{u-1}=1$ for every $k\in[\log u]$ by definition and the assumption that $u$ is a power of 2. Therefore we have
\begin{equation}
    \log u=\sum_{k=1}^{\log u}1=\sum_{k=1}^{\log u}c^{(k)}_{u-1}.
    \label{eq: inlemma: lg u}
\end{equation}
Furthermore, we have $c^{(k)}_{u-1} = \max_{j\in[x_m, x_{m+1})} c^{(k)}_j = \max_{[x_m + a, x_{m + 1} + a)} c^{(k)}_{j}$, using the assumption that $x_{m} + a < u \le x_{1}+u=x_{m+1}$ as $S_{\ge}=\emptyset$.
From Definition~\ref{def:trie prefix-free}, Equations~\eqref{eq: trie x y = summax}~and~\eqref{eq: inlemma: lg u}, we then obtain
\begin{align}
    \trie(S+a)
    &=\log u + \sum_{k=1}^{\log u}\sum_{i=1}^{m-1} \max_{j \in [x_{i} + a, x_{i+1} + a)} c^{(k)}_{j}\notag\\
    &=\sum_{k=1}^{\log u}\left( 
        \max_{j\in [x_m + a, x_{m+1} + a)} c^{(k)}_{j} + 
        \sum_{i=1}^{m-1} \max_{j\in [x_{i} + a, x_{i+1} + a)} c^{(k)}_{j}
        \right)\notag\\
    &=\sum_{k=1}^{\log u}\sum_{i=1}^m \max_{j\in [x_{i} + a, x_{i+1} + a)} c^{(k)}_{j}.
    \label{eq: inlemma: trie S+a}
\end{align}
The claim follows analogously when $S_{<}=\emptyset$, since then $c^{(k)}_j=c^{(k)}_{j+u}$ for every $k\in[\log u]$ and every $j\ge 0$.

Now we assume that both $S_{<}$ and $S_{\ge}$ are not empty. Observe that $x_1\in S_{<}$ because it is the smallest element in $S$ and $S_{<}$ is not empty. Recalling that $x_{m+1}:=x_{1}+u$, we can observe that the trie for $(S_{\ge}\cup\{x_{m+1}-u\})+a$ and the trie for $S_{<}+a$ share exactly one path from the root to $x_1+a$. Therefore,
\begin{equation}
    \trie(S+a)=\trie((S_{\ge}\cup\{x_{m+1}-u\})+a)+\trie(S_{<}+a)-\log u.
    \label{eq: inlemma: shifted}
\end{equation}
Let $i^*\in[2, m]$ be the integer such that $x_{i^*}=\min S_{\ge}$. Then $S_<=\{x_1,\cdots,x_{i^*-1}\}$ and $S_{\ge}\cup\{x_{m+1}\}=\{x_{i^*}-u,x_{i^*+1}-u,\cdots,x_{m+1}-u\}$. 
Since $0\le x_1+a <u \le x_{i*}+a$, for every $k\in[\log u]$ it holds that 
$\max_{j\in[ x_1+a , x_{i*}+a)} c^{(k)}_j = c^{(k)}_{u-1} = 1$. 
Therefore we have
\begin{align*}
    \trie(S_{<}+a)&=\log u +\sum_{k=1}^{\log u}\sum_{i=1}^{i^*-2} \max_{j \in [x_i + a, x_{i+1} + a)} c^{(k)}_{j} 
    =\sum_{k=1}^{\log u}\sum_{i=1}^{i^*-1} \max_{j\in [x_i + a, x_{i+1} + a)} c^{(k)}_{j}
\end{align*}
Applying this with Equation~\eqref{eq: inlemma: trie S+a} to Equation~\eqref{eq: inlemma: shifted}, we obtain
\begin{align*}
    \trie(S+a) &= \sum_{k=1}^{\log u}\left (\sum_{i=i^*}^{m+1} \max_{j\in[x_{i} + a, x_{i+1} + a)} c^{(k)}_{j} + \sum_{i=1}^{i^*-1} \max_{j\in[x_{i} + a, x_{i+1} + a)} c^{(k)}_{j}\right )-\log u\\
    &= \sum_{k=1}^{\log u} \sum_{i=1}^m \max_{j\in[x_{i} + a, x_{i+1} + a)} c^{(k)}_{j}+\sum_{k=1}^{\log u} \max_{j\in [x_{m+1} + a, x_{m+2} + a)} c^{(k)}_{j}-\log u.
\end{align*}
where we define $x_{m+2}:=x_{i^*}+u$.

Now note that
$x_{m+1}+a=x_1+a+u<2u=u+u\le x_{i^*}+a+u=x_{m+2}+a$ and thus
$x_{m+1}+a\le 2u-1<x_{m+2}+a$. Since $c^{(k)}_{2u-1}=1$ by definition, it follows that the second maximum is always 1. Hence the second term equals $\log u$, which is canceled with the last term, and thus the claim follows.
\end{proof}


\subsection{Implementation of \texorpdfstring{$\mathcal D$}{D} for the \texorpdfstring{$O(u+N\log u)$}{O(u+N log u)}-time algorithm}
\label{app: deferred u+Nlogu}
All functions implementing $\mathcal{D}$ besides the $\mathtt{extend}$ function are self-explanatory. 
Thus, we proceed with a detailed explanation of $\mathcal{D}.\mathtt{extend}()$, which corresponds to duplicating the sequence $A$ to $A'=AA$. Assume that $A$ is of length $i$, i.e., contains elements $A[0], \ldots, A[i - 1]$ and recall that $\Delta$ at every point has to encode $A$ via differences, i.e., $\Delta[0]=A[0]$ and $\Delta[a] = A[a] - A[a - 1]$ for any $a\in [1, i)$. Hence, when duplicating the sequence $A$, we can simply duplicate $\Delta$ to $\Delta'=\Delta\Delta$ but have to change the value at the ``boundary'', i.e., $\Delta'[i]$. This value has to become $A'[i] - A'[i - 1] = A[0] - A[i - 1]$. We can simply obtain $A[0]$ from $\Delta[0]$ and we can reconstruct $A[i - 1]$ as $\sum_{a = 0}^{i - 1}\Delta[a]$.

\begin{algorithm}[t]
    \SetKwFunction{FInit}{$\mathcal{D}.\mathtt{initialize}$}
    \SetKwFunction{FAdd}{$\mathcal{D}.\mathtt{add}$}
    \SetKwFunction{FExtend}{$\mathcal{D}.\mathtt{extend}$}
    \SetKwFunction{FArgmin}{$\mathcal{D}.\mathtt{argmin}$}
    \SetKwProg{Fn}{function}{:}{}
    
    \Fn{\FInit{}}{
        $\Delta \gets \langle 0 \rangle$\;
    }
    \Fn{\FAdd{$\ell$, $r$}}{
        $\Delta[\ell] \gets \Delta[\ell]+1$ \tcp*{equivalent to $A[a]\gets A[a]+1,\forall a\in[\ell,r)$} \label{alg: u+Nlogu line: add1}
        \lIf{$r<|\Delta|$}{$\Delta[r] \gets \Delta[r]-1$} \label{alg: u+Nlogu line: add2}
    }
    \Fn{\FExtend{}}{
        $i\gets|\Delta|$; $s\gets \sum_{a\in[0,i)}\Delta[a]$ \tcp*{$s$ becomes $A[i-1]$}
        $\Delta\gets\Delta\Delta$;
        $\Delta[i]\gets\Delta[i]-s$ \tcp*{$\Delta[i]$ becomes $A[0] - A[i-1]$} 
    }
    \Fn{\FArgmin{}}{
        $a^*\gets 0$; $v^*\gets \Delta[0]$; $v\gets \Delta[0]$ \tcp*{finding the minimum in the prefix sum}
        \For{$a = 1 .. |\Delta|-1$}{
            $v\gets v+\Delta[a]$\;
            \lIf{$v^*>v$} {
                $a^*\gets a$;
                $v^*\gets v$;
            }
        }
        \textbf{return} $a^*$\;
    }
\caption{Implementation of $\mathcal{D}$ for the $O(u+N\log u)$-time algorithm.} \label{alg: u+Nlogu}
\end{algorithm}

\subsection{Implementation of \texorpdfstring{$\mathcal D$}{D} for the \texorpdfstring{$O(N\log^2u)$}{O(N log2 u)}-time algorithm}
\label{app: deferred Nlog2u}

The pseudocode implementing the interface of the DAG-compressed variant of dynamic segment trees that we described in Section~\ref{subsubsection: Nlog2u} is given in Algorithm~\ref{alg: Nlog2u 2}. 
Algorithm \ref{alg: Nlog2u 2 private functions} describes instead the private functions called by Algorithm~\ref{alg: Nlog2u 2}.
In what follows, we give all implementation details. 

\begin{algorithm}[H]
    \SetKwFunction{FInit}{$\mathcal{D}.\mathtt{initialize}$}
    \SetKwFunction{FAdd}{$\mathcal{D}.\mathtt{add}$}
    \SetKwFunction{FExtend}{$\mathcal{D}.\mathtt{extend}$}
    \SetKwFunction{FArgmin}{$\mathcal{D}.\mathtt{argmin}$}
    
    \SetKwProg{Fn}{function}{:}{}
    
    \Fn{\FInit{}}{
        $\mathcal{D}.\mathtt{root} \gets $\texttt{reallocate\_node}$(\mathrm{null},\mathrm{null},\mathrm{null})$\;
        $\mathcal{D}.\mathtt{height} \gets 1$ \label{alg: line: init height}\;
    }
    \Fn{\FAdd{$\ell$, $r$}}{$\mathcal{D}.\mathtt{root}\gets\mathtt{increment}(\mathcal{D}.\mathtt{root},\ell,r,\mathcal{D}.\mathtt{height})$ \label{alg3: line: add}\;
    }
    \Fn{\FExtend{}}{
    \
        $u\gets \mathcal{D}.\mathtt{root}$;
        $\mathcal{D}.\mathtt{root}\gets $\texttt{reallocate\_node}$(u,u,u)$\;
        $\mathcal{D}.\mathtt{root}.\mathtt{val}\gets0$ ;
        $\mathcal{D}.\mathtt{height}\gets \mathcal{D}.\mathtt{height}+1$\;
    }
    \Fn{\FArgmin{}}{
        \Return $\mathcal{D}.\mathtt{root}.\mathtt{argmin}$\;
    }

\caption{Implementation of $\mathcal{D}$ for the $O(N\log^2 u)$-time algorithm.} \label{alg: Nlog2u 2}
\end{algorithm}

The structure of the DAG-compressed tree is as follows.
It is built in $\log u$ iterations. At each iteration $k\in[\log u]$, the height of the DAG (i.e. the height of the tree resulting from the expansion of the DAG) is $k$ and it is equal to the number of nodes on the path from the root to the leaves. 
At any point, the expansion of the DAG we are building is a complete binary tree of height $k$.
We store this height $k$ in a variable  
$\mathcal{D}.\mathtt{height}$. The root node is stored in $\mathcal{D}.\mathtt{root}$, and it covers the range $[0,2^{k-1})$. We say that the height of the leaves in the tree is 1, parents of leaves are at height 2, and so on. A node $v$ at height $h\ge 1$ covers a range $[x,x+2^{h-1})$ for some integer $x$. Such a node $v$ may have zero or two children. If it has two children, the left child covers $[x,x+2^{h-2})$ and the right child covers $[x+2^{h-2},x+2^{h-1})$, respectively. Note that in the pseudocodes, neither the value of $x$ nor the value of $h$ are stored in the node $v$ explicitly, but they can be reconstructed while navigating the tree.

Node $v$, corresponding to the subsequence $A[x,x+2^{h-1}-1]$, has four variables $v.\mathtt{val}$, $v.\mathtt{argmin}$, $v.\mathtt{min}$, and $v.\mathtt{ref}$, in addition to the pointers to its left and right child $v.\mathtt{left}$ and $v.\mathtt{right}$. 
Recall that the purpose of $\mathcal D$ is to maintain an integer array $A$, increment by one unit all entries belonging to range $A[\ell,r-1]$ via operation $\mathcal D.\mathtt{add}(\ell,r)$, and duplicate the array via operation $\mathcal{D}.\mathtt{extend}()$. For a node $v$ that corresponds to a range $A[x,x+2^{h-1}-1]$, we store those increments that apply to the whole range in a variable $v.\mathtt{val}$. 
The variable $v.\mathtt{argmin}$ stores 
$\argmin_{a\in [0,2^{h-1})} A[x+a]$, while $v.\mathtt{min}$ stores 
$\min_{a\in [0,2^{h-1})} A[x+a]$. 

Finally, $v.\mathtt{ref}$ stores the number of pointers of other nodes to the node $v$. 
The role of this reference counter is to duplicate nodes only when necessary (more details are given below).
This reference counter $v.\mathtt{ref}$ is managed in function $\mathtt{reallocate\_node}(\cdot)$ in Lines~\ref{alg: line: newnode begin}-\ref{alg: line: newnode end}. This function takes three arguments $u$, $w_L$ and $w_R$ to create a new node. If $u\neq\mathrm{null}$, it makes a copy of $u$, and decrements $u.\mathtt{ref}$ because it means one pointer will replace $u$ with its copy. Then it sets the left and right child of the new node to $w_L$ and $w_R$, and (if they are not null pointers) increment $w_L.\mathtt{ref}$ and $w_R.\mathtt{ref}$ by 1 each because the new node will point to them.

All functions implementing $\mathcal{D}$ besides $\mathcal{D}.\mathtt{extend}()$ and $\mathcal{D}.\mathtt{add}(\ell,r)$ are self-explanatory so we do not discuss them.
Function $\mathcal{D}.\mathtt{extend}()$ is also simple as it just allocates a new root, sets its left and right children pointers to the old root, and sets its height to the height of the old root incremented by one unit. Note that the counter $\mathtt{ref}$ associated with the old root is correctly set at 2. The counter $\mathtt{ref}$ associated with the new root is set at 1 (even though the new root is not referenced by any node, this value prevents the code from re-allocating the new root each time function $\mathcal{D}.\mathtt{add}(\ell,r)$ is called).

\begin{algorithm}[t]

    \SetKwFunction{FIncr}{increment}
    \SetKwFunction{FNewnode}{reallocate\_node}

    \SetKwProg{Fn}{function}{:}{}
    
    \Fn{\FNewnode{$u$,$w_L$,$w_R$}}{
        \label{alg: line: newnode begin}
        
        Create a new node $v$\;
        \uIf{$u\ne\mathrm{null}$} { 
            $u.\mathtt{ref}\gets u.\mathtt{ref}-1$ \;
            $(v.\mathtt{min},v.\mathtt{argmin},v.\mathtt{val}) \gets (u.\mathtt{min},u.\mathtt{argmin},u.\mathtt{val})$ \;
            
        }\lElse{
            $(v.\mathtt{val},v.\mathtt{min},v.\mathtt{argmin})\leftarrow(0,0,0)$ 
        }
        $(v.\mathtt{left},v.\mathtt{right})\gets (w_L,w_R)$ \;
            
        \lIf{$w_L\neq \mathrm{null}$} { $w_L.\mathtt{ref}\gets w_L.\mathtt{ref}+1$
        }
        \lIf{$w_R\neq \mathrm{null}$} { $w_R.\mathtt{ref}\gets w_R.\mathtt{ref}+1$
        }
        
        $v.\mathtt{ref}\gets 1$ \;
        \Return $v$ 
        \label{alg: line: newnode end}
    }
    \Fn{\FIncr{$v$,$\ell$,$r$,$h$}}{ \label{alg: line: func increment begin}
        \lIf{$\ell\ge r$}{\Return $v$}
        
        \lIf{$v.\mathtt{ref}>1$}{$v\gets$\texttt{reallocate\_node}$(v,v.\mathtt{left},v.\mathtt{right})$ \label{algo: line: copy node}} 
        
        \If{$r-\ell=2^{h-1}$}{ \label{algo: line: update val}
            $(v.\mathtt{val},~~v.\mathtt{min})\gets (v.\mathtt{val}+1,~~v.\mathtt{min}+1)$\; \Return $v$ \label{algo: line: update val end}
        }
        $v.\mathtt{left}\gets\mathtt{increment}(v.\mathtt{left}, \ell, \min\{r,2^{h-2}\},h-1)$ \label{algo: line: increment recursive left} \;
        $v.\mathtt{right}\gets\mathtt{increment}(v.\mathtt{right}, \max\{\ell-2^{h-2},0\}, r-2^{h-2}, h-1)$ \label{algo: line: increment recursive right}\;
        \uIf{$v.\mathtt{left}.\mathtt{min} \le v.\mathtt{right}.\mathtt{min}$} { \label{algo: line: update min begin}
            $(v.\mathtt{min},~v.\mathtt{argmin})\gets (v.\mathtt{left}.\mathtt{min}+v.\mathtt{val},~v.\mathtt{left}.\mathtt{argmin})$\;
        }
        \lElse{
            $(v.\mathtt{min},~v.\mathtt{argmin})\gets( v.\mathtt{right}.\mathtt{min}+v.\mathtt{val},~v.\mathtt{right}.\mathtt{argmin}+2^{h-2})$
            \label{algo: line: update min end}
        }
        \Return $v$ \label{alg: line: func increment end}
    }
\caption{Private functions of the data structure from Algorithm~\ref{alg: Nlog2u 2}.}
\label{alg: Nlog2u 2 private functions}
\end{algorithm}

We proceed by discussing the details of $\mathcal{D}.\mathtt{add}(\ell,r)$. 
When this function is invoked, we call a  subroutine $\mathtt{increment}(\cdot)$ with arguments representing the root node, the interval $[\ell,r)$, and the height of the tree (Line~\ref{alg3: line: add}).
We start from the root node, and perform $\mathtt{increment}(\cdot)$ recursively. Suppose we arrive at a node $v$ of height $h$ covering a range $[x,x+2^{h-1})$. In Line~\ref{algo: line: copy node}, if $v$ is referred by more than one pointer (i.e., if $v.\mathtt{ref}>1$), this means that the range $[x,x+2^{h-1})$ is not the unique range that $v$ is covering; in other words, $v$ is being reused at more than one place. Thus we make a copy of $v$, and proceed with the copied node. This new node will be returned by the function (Lines~\ref{algo: line: update val end},\ref{alg: line: func increment end}) so that it can replace the old node properly (Lines~\ref{alg3: line: add},\ref{algo: line: increment recursive left}-\ref{algo: line: increment recursive right}). Now we are to perform an update according to the given interval. When $[\ell,r)$ is passed as an argument of the function, it means that we are to increment $A[a]$ for $a\in[x+\ell,x+r)$. If the size of the interval is exactly $2^{h-1}$, we need to increase $A[a]$ by 1 unit for all $a\in[x,x+2^{h-1})$, which can be performed by incrementing $v.\mathtt{val}$ and $v.\mathtt{min}$ by 1 each (Lines~\ref{algo: line: update val}-\ref{algo: line: update val end}). Otherwise, we split the interval at position $x+2^{h-2}$ (i.e., split $[x+\ell,x+r)$ into $[x+\ell,x+2^{h-2})$ and $[x+2^{h-2},x+r)$), then process each split segment with its left and right child, respectively (Lines~\ref{algo: line: increment recursive left}-\ref{algo: line: increment recursive right}).
After processing the insertion at the child nodes, we update $v.\mathtt{min}$ and $v.\mathtt{argmin}$ according to the minimum computed in the children (Lines~\ref{algo: line: update min begin}-\ref{algo: line: update min end}), which will propagate to the root node. If the minimum is from the right child, we add the offset $2^{h-2}$ for updating $v.\mathtt{argmin}$ properly. 

The cost of procedure $\mathcal{D}.\mathtt{add}(\cdot)$ is proportional to $O(\log|A|)$. 
This is because in Line \ref{algo: line: update val} we do not further recurse on the children of $v$ if the local interval $[\ell,r)$ spans the entire range of length $2^{h-1}$ associated with $v$; in turn, this means that the recursive calls to $\mathtt{increment}$ stop on the $O(\log |A|)$ nodes whose associated intervals partition the initial range $[\ell,r)$ on which function $\mathcal{D}.\mathtt{add}(\ell,r)$ was called. Additionally, all ancestors of those nodes have at least two children (notice that, if $\mathtt{increment}$ is called recursively, the recursion always occurs on both children of the node: see Lines \ref{algo: line: increment recursive left} and \ref{algo: line: increment recursive right}), therefore the total number of nodes recursively visited by $\mathcal{D}.\mathtt{add}(\ell,r)$ is  $O(\log |A|)$.

\section{Additional Material for Section~\ref{sec:experiments}}
\begin{table}[h]
\centering
\renewcommand{\arraystretch}{0.9}
\begin{tabular}{r|l|rrrrr}\hline
\multicolumn{2}{l|}{dataset}                        & $\sigma$      & $u$ & $N$      & $n$ & $N/n$  \\ \hline   
1&email-Enron-core-seqs      & 141      & 256       & 14148    & 10428 & 1.36 \\
2&contact-prim-school-seqs   & 242      & 256       & 251546   & 174796 & 1.44 \\
3&contact-high-school-seqs   & 327      & 512       & 377000   & 308990 & 1.22 \\
4&email-Eu-core-seqs         & 937      & 1024      & 252872   & 202769 & 1.25 \\
5&tags-mathoverflow-seqs     & 1399     & 2048      & 125056   & 44950 & 2.78  \\
6&tags-math-sx-seqs          & 1650     & 2048      & 1177312  & 517810 & 2.27 \\
7&coauth-Business-seqs       & 236226   & 2097152   & 849838   & 463070 & 1.84 \\  
8&coauth-Geology-seqs     & 525348     & 2097152 & 3905349   & 1438652 & 2.71  \\ \hline
9&DBLP.xml                   & 26       & 32        & 28815437 & 3939813 & 7.31  \\ \hline
10&proteins-high-id           & 20       & 32        & 40664    &  2128   & 19.11 \\
11&proteins-long              & 25       & 32        & 11051828 &  571282 & 19.35 \\\hline
12&book-ratings               & 5        & 8         & 48763    &  10000  & 4.88 \\
13&tags-book                  & 34252    & 65536     & 999904   &  10000  & 99.99 \\ \hline
\end{tabular}
\caption{Summary of the thirteen datasets. From left to right, we report the dataset id and name, the number of distinct universe elements belonging to the sets ($\sigma$), the size of the universe ($u = 2^{\lceil \log |U| \rceil}$), the total length ($N$), the number of sets ($n$), and the average size of a set ($N/n$).}
\label{tab:datasets}
\end{table}
\begin{table}[h]
\centering
\renewcommand{\arraystretch}{0.9}
\setlength{\tabcolsep}{4.15pt}
\begin{tabular}{r|r|rr|rr|rr}\hline
 \multirow{2}{*}{id}                       & \multirow{2}{*}{opt-shift}  &  \multirow{2}{*}{avg-shift} & $ \multirow{2}{*}{\smash{$\displaystyle \frac{\textrm{opt-shift}}{\textrm{avg-shift(\%)}}$}} $ & \multirow{2}{*}{worst-shift}  & \multirow{2}{*}{\smash{$\displaystyle \frac{\textrm{opt-shift}}{\textrm{worst-shift(\%)}}$}} & \multirow{2}{*}{opt-ord} & \multirow{2}{*}{\smash{$\displaystyle \frac{\textrm{opt-shift}}{\textrm{ opt-ord(\%) }}$}} \\
&&&&&&& \\ \hline
1      & 105708   &  106787   &      98.99 & 108027    & 97.85 & 84843 & 124.59 \\
2   & 1864240  &  1870395  &      99.67 & 1876731   & 99.33 & 1815243 & 102.70 \\
3   & 3227940  &  3244486  &      99.49 & 3257990   & 99.07 & 2906848 & 111.05 \\
4         & 2457044  &  2460375  &      99.86 & 2464644   & 99.69 & 2181173 & 112.65 \\
5     & 1092079  &  1138803  &      95.90 & 1177890   & 92.71 & 880232 & 124.07 \\
6          & 10727737 &  11069006 &      96.92 & 11413439  & 93.99 & 8426162 & 127.31  \\
7       & 14940827 &  15073904 &      99.12 & 15171942  & 98.47 &-&- \\
8     & 66864377  &  67972793  & 98.37 & 68812124 & 97.17 &-&- \\ \hline
9 & 70403690 & 75662919 & 93.05& 81715980 & 86.15 & 58285023 & 120.79 \\ \hline
10 & 83178 & 86608 & 96.04& 89503 & 92.93 & 80989 & 102.70 \\
11 & 22454406 & 23371923 & 96.07& 24144585 & 93.00 & 21868947 & 102.68 \\ \hline
12 & 96467 & 97865 & 98.57& 98714 & 97.80 & 87566 & 110.16 \\
13 & 7923290 & 8022933 & 98.76& 8146000 & 97.27 &  - & - \\ \hline
\end{tabular}
\caption{
Experimental result. From left to right, we report the dataset id, and the optimal shifted encoding size (opt-shift) followed by the average shifted encoding size (avg-shift), the ratio between opt-shift and avg-shift, the worst shifted encoding size (worst-shift), and the ratio between opt-shift and worst-shift. The last two columns show the optimal shifted ordered encoding size (opt-ord), and the ratio between opt-shift and opt-ord (only for datasets with small universe size $u$). }
\label{tab:summary}
\end{table}

\end{document}

%% file: macros.tex
\DeclareMathOperator{\argmin}{argmin}

\DeclareMathOperator{\trie}{trie}
\DeclareMathOperator{\enc}{enc}

\DeclareMathOperator{\modulo}{mod}

\let\epsilon\varepsilon

\usepackage{dsfont}
\newcommand{\ones}{\mathds{1}}

\usepackage{xspace}

\usepackage[linesnumbered,ruled,vlined,boxed]{algorithm2e}
\newlength{\commentWidth}
\let\oldnl\nl
\newcommand{\nonl}{\renewcommand{\nl}{\let\nl\oldnl}}
\DontPrintSemicolon
\SetKwInOut{Input}{Input}\SetKwInOut{Output}{Output}

\usepackage{dashrule}

\usepackage{tikz}
\usetikzlibrary{decorations.pathreplacing}
\usetikzlibrary{plotmarks}
\usetikzlibrary{positioning,automata,arrows}
\usetikzlibrary{shapes.geometric}
\usetikzlibrary{decorations.markings}
\usetikzlibrary{positioning, shapes, arrows}

\PassOptionsToPackage{usenames,dvipsnames,svgnames}{xcolor}
\definecolor{orange}{RGB}{235,90,0}
\definecolor{darkorange}{RGB}{175,30,0}
\definecolor{turkis}{RGB}{131,182,182}
\definecolor{darkturkis}{RGB}{31,82,82}
\definecolor{green}{RGB}{102,180,0}
\definecolor{darkgreen}{RGB}{51,90,0}
\definecolor{myblue}{RGB}{0,0,213}
\definecolor{mydarkblue}{RGB}{0,0,100}
\definecolor{mybrightblue}{HTML}{74B0E4}
\definecolor{mybrighterblue}{HTML}{B3EAFA}
\definecolor{lila}{RGB}{102,0,102}
\definecolor{darkred}{RGB}{139,0,0}
\definecolor{darkyellow}{RGB}{188,135,2}
\definecolor{brightgray}{RGB}{200,200,200}
\definecolor{darkgray}{RGB}{50,50,50}
\definecolor{amaranth}{rgb}{0.9, 0.17, 0.31}
\definecolor{alizarin}{rgb}{0.82, 0.1, 0.26}
\definecolor{amber}{rgb}{1.0, 0.75, 0.0}
\definecolor{green(ryb)}{rgb}{0.4, 0.69, 0.2}
\definecolor{hanblue}{rgb}{0.27, 0.42, 0.81}
\definecolor{grannysmithapple}{rgb}{0.66, 0.89, 0.63}

%% file: figures/organ-revised.tex
\begin{tikzpicture}[
x=0.5cm,y=-0.4cm,
c/.style={circle, inner sep=0pt, minimum size=3pt},
b/.style={draw=black,fill=black},
b1/.style={draw=black,line width=0.4mm},
b2/.style={draw=hanblue,line width=0.4mm},
b3/.style={draw=teal,line width=0.4mm},
b4/.style={draw=magenta,line width=0.4mm},
g/.style={draw=lightgray, fill=lightgray},
pink1/.style={line width=0.25mm,fill=cyan},
pink2/.style={line width=0.25mm,fill=grannysmithapple},
pink3/.style={line width=0.25mm,fill=pink},
pink4/.style={line width=0.25mm,fill=lightgray},
]
\tikzstyle{every node}=[font=\small]
\clip (-2,-0.5) rectangle (19.5,9.5);

\node [c,b] (v00) at ( 0,0) {};

\node [c,b] (v10) at ( 0,1) {};
\node [c,b] (v11) at ( 8,1) {};

\node [c,b] (v20) at ( 0,2) {};
\node [c,b] (v21) at ( 4,2) {};
\node [c,b] (v22) at ( 8,2) {};
\node [c,b] (v23) at (12,2) {};

\node [c,g] (v30) at ( 0,3) {};
\node [c,b] (v31) at ( 2,3) {};
\node [c,b] (v32) at ( 4,3) {};
\node [c,g] (v33) at ( 6,3) {};
\node [c,g] (v34) at ( 8,3) {};
\node [c,b] (v35) at (10,3) {};
\node [c,b] (v36) at (12,3) {};
\node [c,g] (v37) at (14,3) {};

\node [c,g] (v40) at ( 0,4) {};
\node [c,g] (v41) at ( 1,4) {};
\node [c,b] (v42) at ( 2,4) {};
\node [c,g] (v43) at ( 3,4) {};
\node [c,b] (v44) at ( 4,4) {};
\node [c,g] (v45) at ( 5,4) {};
\node [c,g] (v46) at ( 6,4) {};
\node [c,g] (v47) at ( 7,4) {};
\node [c,g] (v48) at ( 8,4) {};
\node [c,g] (v49) at ( 9,4) {};
\node [c,b] (v4a) at (10,4) {};
\node [c,g] (v4b) at (11,4) {};
\node [c,g] (v4c) at (12,4) {};
\node [c,b] (v4d) at (13,4) {};
\node [c,g] (v4e) at (14,4) {};
\node [c,g] (v4f) at (15,4) {};

\draw [b1] (v00) -- (v10);
\draw [b3] (v00) -- (v11);

\draw [b1] (v10) -- (v20);
\draw [b2] (v10) -- (v21);
\draw [b3] (v11) -- (v22);
\draw [b4] (v11) -- (v23);

\draw [g] (v20) -- (v30);
\draw [b1] (v20) -- (v31);
\draw [b2] (v21) -- (v32);
\draw [g] (v21) -- (v33);
\draw [g] (v22) -- (v34);
\draw [b3] (v22) -- (v35);
\draw [b4] (v23) -- (v36);
\draw [g] (v23) -- (v37);

\draw [g] (v30) -- (v40);
\draw [g] (v30) -- (v41);
\draw [b1] (v31) -- (v42);
\draw [g] (v31) -- (v43);
\draw [b2] (v32) -- (v44);
\draw [g] (v32) -- (v45);
\draw [g] (v33) -- (v46);
\draw [g] (v33) -- (v47);
\draw [g] (v34) -- (v48);
\draw [g] (v34) -- (v49);
\draw [b3] (v35) -- (v4a);
\draw [g] (v35) -- (v4b);
\draw [g] (v36) -- (v4c);
\draw [b4] (v36) -- (v4d);
\draw [g] (v37) -- (v4e);
\draw [g] (v37) -- (v4f);

\node [c,g] (u00) at (16,0) {};
\node [c,g] (u10) at (16,1) {};
\node [c,g] (u11) at (24,1) {};
\node [c,g] (u20) at (16,2) {};
\node [c,g] (u21) at (20,2) {};
\node [c,g] (u30) at (16,3) {};
\node [c,g] (u31) at (18,3) {};
\node [c,g] (u40) at (16,4) {};
\node [c,g] (u41) at (17,4) {};
\node [c,g] (u42) at (18,4) {};
\node [c,g] (u43) at (19,4) {};
\draw [g] (u00) -- (u10);
\draw [g] (u00) -- (u11);
\draw [g] (u10) -- (u20);
\draw [g] (u10) -- (u21);
\draw [g] (u20) -- (u30);
\draw [g] (u20) -- (u31);
\draw [g] (u30) -- (u40);
\draw [g] (u30) -- (u41);
\draw [g] (u31) -- (u42);
\draw [g] (u31) -- (u43);

\draw [    ] ( 2-0.35,6-0.35) rectangle ++ (2-0.3,0.7);
\draw [pink2] ( 4-0.35,6-0.35) rectangle ++ (6-0.3,0.7);
\draw [    ] (10-0.35,6-0.35) rectangle ++ (3-0.3,0.7);
\draw [pink4] (13-0.35,6-0.35) rectangle ++ (5-0.3,0.7);

\draw [pink1] ( 2-0.35,7-0.35) rectangle ++ (2-0.3,0.7);
\draw [pink2] ( 4-0.35,7-0.35) rectangle ++ (6-0.3,0.7);
\draw [pink3] (10-0.35,7-0.35) rectangle ++ (3-0.3,0.7);
\draw [pink4] (13-0.35,7-0.35) rectangle ++ (5-0.3,0.7);

\draw [pink1] ( 2-0.35,8-0.35) rectangle ++ (2-0.3,0.7);
\draw [pink2] ( 4-0.35,8-0.35) rectangle ++ (6-0.3,0.7);
\draw [pink3] (10-0.35,8-0.35) rectangle ++ (3-0.3,0.7);
\draw [pink4] (13-0.35,8-0.35) rectangle ++ (5-0.3,0.7);

\draw [pink1] ( 2-0.35,9-0.35) rectangle ++ (2-0.3,0.7);
\draw [pink2] ( 4-0.35,9-0.35) rectangle ++ (6-0.3,0.7);
\draw [pink3] (10-0.35,9-0.35) rectangle ++ (3-0.3,0.7);
\draw [pink4] (13-0.35,9-0.35) rectangle ++ (5-0.3,0.7);

\node at (-1,5) {$c$};
\node at ( 0,5) {1};
\node at ( 1,5) {2};
\node at ( 2,5) {1};
\node at ( 3,5) {3};
\node at ( 4,5) {1};
\node at ( 5,5) {2};
\node at ( 6,5) {1};
\node at ( 7,5) {4};
\node at ( 8,5) {1};
\node at ( 9,5) {2};
\node at (10,5) {1};
\node at (11,5) {3};
\node at (12,5) {1};
\node at (13,5) {2};
\node at (14,5) {1}; 
\node at (15,5) {4}; 
\node at (16,5) {1}; 
\node at (17,5) {2}; 
\node at (18,5) {1}; 
\node at (19,5) {3}; 

\node at (-1,9) {$c^{(1)}$};
\node at ( 0,9) {1};
\node at ( 1,9) {1};
\node at ( 2,9) {1};
\node at ( 3,9) {1};
\node at ( 4,9) {1};
\node at ( 5,9) {1};
\node at ( 6,9) {1};
\node at ( 7,9) {1};
\node at ( 8,9) {1};
\node at ( 9,9) {1};
\node at (10,9) {1};
\node at (11,9) {1};
\node at (12,9) {1};
\node at (13,9) {1};
\node at (14,9) {1}; 
\node at (15,9) {1}; 
\node at (16,9) {1}; 
\node at (17,9) {1}; 
\node at (18,9) {1}; 
\node at (19,9) {1}; 

\node at (-1,8) {$c^{(2)}$};
\node at ( 0,8) {0};
\node at ( 1,8) {1};
\node at ( 2,8) {0};
\node at ( 3,8) {1};
\node at ( 4,8) {0};
\node at ( 5,8) {1};
\node at ( 6,8) {0};
\node at ( 7,8) {1};
\node at ( 8,8) {0};
\node at ( 9,8) {1};
\node at (10,8) {0};
\node at (11,8) {1};
\node at (12,8) {0};
\node at (13,8) {1};
\node at (14,8) {0}; 
\node at (15,8) {1}; 
\node at (16,8) {0}; 
\node at (17,8) {1}; 
\node at (18,8) {0}; 
\node at (19,8) {1}; 

\node at (-1,7) {$c^{(3)}$};
\node at ( 0,7) {0};
\node at ( 1,7) {0};
\node at ( 2,7) {0};
\node at ( 3,7) {1};
\node at ( 4,7) {0};
\node at ( 5,7) {0};
\node at ( 6,7) {0};
\node at ( 7,7) {1};
\node at ( 8,7) {0};
\node at ( 9,7) {0};
\node at (10,7) {0};
\node at (11,7) {1};
\node at (12,7) {0};
\node at (13,7) {0};
\node at (14,7) {0}; 
\node at (15,7) {1}; 
\node at (16,7) {0}; 
\node at (17,7) {0}; 
\node at (18,7) {0}; 
\node at (19,7) {1}; 

\node at (-1,6) {$c^{(4)}$};
\node at ( 0,6) {0};
\node at ( 1,6) {0};
\node at ( 2,6) {0};
\node at ( 3,6) {0};
\node at ( 4,6) {0};
\node at ( 5,6) {0};
\node at ( 6,6) {0};
\node at ( 7,6) {1};
\node at ( 8,6) {0};
\node at ( 9,6) {0};
\node at (10,6) {0};
\node at (11,6) {0};
\node at (12,6) {0};
\node at (13,6) {0};
\node at (14,6) {0}; 
\node at (15,6) {1}; 
\node at (16,6) {0}; 
\node at (17,6) {0}; 
\node at (18,6) {0}; 
\node at (19,6) {0}; 

\end{tikzpicture}

%% file: figures/costinterval-a.tex
\begin{tikzpicture}[
x=0.5cm,y=-0.4cm,
r/.style={fill=pink},
rc/.style={circle,draw=red, fill=red, inner sep=0pt, minimum size=4pt},
rs/.style={rectangle,draw=red, inner sep=0pt, minimum size=6pt},
bc/.style={circle,draw=blue, inner sep=0pt, minimum size=4pt},
kc/.style={circle,draw=black, fill=black, inner sep=0pt, minimum size=3pt},
]

\def\x{2};
\def\y{4};
\def\k{0};\draw [r] ( \x,\k-0.4) rectangle (\y-0.1,\k+0.4);
\def\k{1};\draw [r] ( \x,\k-0.4) rectangle (\y-0.1,\k+0.4);
\def\k{2};\draw [ ] ( \x,\k-0.4) rectangle (\y-0.1,\k+0.4);
\def\k{3};\draw [ ] ( \x,\k-0.4) rectangle (\y-0.1,\k+0.4);
\def\k{4};\draw [r] ( \x,\k-0.4) rectangle (\y-0.1,\k+0.4);

\def\x{4};
\def\y{10};
\def\k{0};\draw [r] ( \x,\k-0.4) rectangle (\y-0.1,\k+0.4);
\def\k{1};\draw [r] ( \x,\k-0.4) rectangle (\y-0.1,\k+0.4);
\def\k{2};\draw [r] ( \x,\k-0.4) rectangle (\y-0.1,\k+0.4);
\def\k{3};\draw [r] ( \x,\k-0.4) rectangle (\y-0.1,\k+0.4);
\def\k{4};\draw [r] ( \x,\k-0.4) rectangle (\y-0.1,\k+0.4);

\def\x{10};
\def\y{13};
\def\k{0};\draw [r] ( \x,\k-0.4) rectangle (\y-0.1,\k+0.4);
\def\k{1};\draw [r] ( \x,\k-0.4) rectangle (\y-0.1,\k+0.4);
\def\k{2};\draw [ ] ( \x,\k-0.4) rectangle (\y-0.1,\k+0.4);
\def\k{3};\draw [r] ( \x,\k-0.4) rectangle (\y-0.1,\k+0.4);
\def\k{4};\draw [r] ( \x,\k-0.4) rectangle (\y-0.1,\k+0.4);

\def\x{13};
\def\y{18};
\def\k{0};\draw [r] ( \x,\k-0.4) rectangle (\y-0.1,\k+0.4);
\def\k{1};\draw [r] ( \x,\k-0.4) rectangle (\y-0.1,\k+0.4);
\def\k{2};\draw [r] ( \x,\k-0.4) rectangle (\y-0.1,\k+0.4);
\def\k{3};\draw [r] ( \x,\k-0.4) rectangle (\y-0.1,\k+0.4);
\def\k{4};\draw [r] ( \x,\k-0.4) rectangle (\y-0.1,\k+0.4);

\node at ( 2+0.5,-2) {$x_1$};
\node at ( 4+0.5,-2) {$x_2$};
\node at (10+0.5,-2) {$x_3$};
\node at (13+0.5,-2) {$x_4$};
\node at (18+0.5,-2) {$x_5$};
\node [anchor=west] at (19,-2) {$(=x_1+u)$};

\node at (-1.0,-1) {$j$};
 \foreach \x in {0,...,18}{
        \node at (\x+0.5,-1) {\x};
    }
\node at (19.7,-1) {$\cdots$};
\draw (-3,-0.5)--(20.5,-0.5);

\foreach \y in {0,...,4} 
{
    \node at (-1.5,\y) {$a=\y:$};
    \foreach \x in {0,...,18}{
        \node at (\x+0.5,\y) {
            \pgfmathparse{int(mod(\x+\y+1,4))}
            \ifthenelse{\pgfmathresult=0}{1}{0}
        };
    }
    \node at (19.7,\y) {$\cdots$};
}

\draw [dotted] (-3,3.5)--(20.5,3.5);
\draw [<->,blue] (20.5,-0.5)--(20.5,3.5);
\node[label={[label distance=0cm,rotate=-90,blue] left: period=$2^{k-1}$}] at (21.4,4.5) {};

\node[label={[label distance=0cm,rotate=-90] left: $\cdots$}] at (-1,6) {};

\node[label={[label distance=0cm,rotate=-90] left: $\cdots$}] at (8.3,6) {};

\end{tikzpicture}

%% file: figures/grammartree1.tex
\begin{tikzpicture}[
x=0.45cm,y=0.5cm,
v/.style={circle, inner sep=0pt, minimum size=4mm,draw=black,fill=white},
l/.style={color=blue},
]
\node    at ( 3.5, 8) {};
 \node[v] (v1) at ( 1.5, 6.2) {2};
  \node[v] (v2) at ( 0.5, 4.4) {6}; \node[l] at ( -0.5, 4.6) {$v_a$};
   \node[v] (v3) at ( 0.0, 1.5) {4}; \node[l] at ( -0.9, 1.1) {$v_c$};
  \node[v] (v4) at ( 2.5, 4.4) {4}; \node[l] at ( 3.5, 4.6) {$v_b$};
   \node[v] (v5) at ( 3.0, 1.5) {5};

\path[->] (v1) edge [bend right=25] (v2);
\path[->] (v1) edge [bend left =25] (v4);
\path[->] (v2) edge [bend right=45] (v3);
\path[->] (v2) edge [bend left =75] (v3);   
\path[->] (v4) edge [bend right=25] (v3);
\path[->] (v4) edge [bend left =45] (v5);  

\draw (-0.5,-0.5) rectangle ++ (4,1);
\foreach \x in {1,...,3}{
    \draw (\x-0.5,-0.5) -- (\x-0.5,0.5);
}
\node at (-1,0) {$A$};
\node at ( 0,0) {12};
\node at ( 1,0) {12};
\node at ( 2,0) {10};
\node at ( 3,0) {11};

\node at ( 2.5,-1.3) {};

\end{tikzpicture}

%% file: figures/grammartree2.tex
\begin{tikzpicture}[
x=0.45cm,y=0.5cm,
v/.style={circle, inner sep=0pt, minimum size=4mm,draw=black,fill=white},
l/.style={color=blue},
]
\node[v] (v0) at ( 3.5, 8) {0}; \node[l] at (4.5, 7.8) {$v_d$};
 \node[v] (v1) at ( 1.5, 6.2) {2}; \node[l] at (0.5, 6.5) {$v_e$};
  \node[v] (v2) at ( 0.5, 4.4) {6};
   \node[v] (v3) at ( 0.0, 1.5) {4};
  \node[v] (v4) at ( 2.5, 4.4) {4};
   \node[v] (v5) at ( 3.0, 1.5) {5};

\path[->] (v0) edge [bend right=45] (v1);
\path[->] (v0) edge [bend left =75] (v1);
\path[->] (v1) edge [bend right=25] (v2);
\path[->] (v1) edge [bend left =25] (v4);
\path[->] (v2) edge [bend right=45] (v3);
\path[->] (v2) edge [bend left =75] (v3);   
\path[->] (v4) edge [bend right=15] (v3);
\path[->] (v4) edge [bend left =45] (v5);  

\draw (-0.5,-0.5) rectangle ++ (8,1);
\foreach \x in {1,...,7}{
    \draw (\x-0.5,-0.5) -- (\x-0.5,0.5);
}
\node at (-1,0) {$A$};
\node at ( 0,0) {12};
\node at ( 1,0) {12};
\node at ( 2,0) {10};
\node at ( 3,0) {11};
\node at ( 4,0) {12};
\node at ( 5,0) {12};
\node at ( 6,0) {10};
\node at ( 7,0) {11};

\node at ( 2.5,-1.3) {} ;

\end{tikzpicture}

%% file: figures/grammartree3-36.tex
\begin{tikzpicture}[
x=0.45cm,y=0.5cm,
v/.style={circle, inner sep=0pt, minimum size=4mm,draw=black,fill=white},
vb/.style={draw=red,text=red},
l/.style={color=blue},
]
\node[v,draw=red] (v0) at ( 3.5, 8) {0};
 \node[v,draw=red] (v1) at ( 1.5, 6.2) {2};
  \node[v] (v2) at ( 0.5, 4.4) {6}; 
   \node[v] (v3) at ( 0.0, 1.5) {4};
  \node[v,draw=red] (v4) at ( 2.5, 4.4) {4}; 
   \node[v,vb] (v5) at ( 3.0, 1.6) {\textbf{6}}; \node[l] at (2.2, 1.8) {$v_g$};
 \node[v,draw=red] (v6) at ( 5.5, 6.2) {2}; 
  \node[v,vb] (v7) at ( 4.5, 4.4) {\textbf{7}}; \node[l] at (4.0, 5.1) {$v_h$};
  \node[v] (v9) at ( 6.5, 4.4) {4};
   \node[v] (v8) at ( 7.0, 1.5) {5}; 
  
\path[->] (v0) edge [bend right=25] (v1);
\path[->] (v0) edge [bend left =25] (v6);
\path[->] (v1) edge [bend right=25] (v2);
\path[->] (v1) edge [bend left =25] (v4);
\path[->] (v2) edge [bend right=45] (v3);
\path[->] (v2) edge [bend left =75] (v3);   
\path[->] (v4) edge [bend right=15] (v3);
\path[->] (v4) edge [bend left =45] (v5);  

\path[->] (v6) edge [bend right=25] (v7);  
\path[->] (v6) edge [bend left =25] (v9);  
\path[->] (v7) edge [to path={(\tikztostart) .. controls +(-1.5,-1.5) and +(2,2.5) ..  (\tikztotarget) \tikztonodes}] (v3);  
\path[->] (v7) edge [to path={(\tikztostart) .. controls +(1.5,-1.5) and +(2,1.8) ..  (\tikztotarget) \tikztonodes}] (v3);  

\path[->] (v9) edge [to path={(\tikztostart) .. controls +(-1.8,-1.8) and +(2,1.2) ..  (\tikztotarget) \tikztonodes}] (v3);  
\path[->] (v9) edge [bend left=15] (v8);  

\draw (-0.5,-0.5) rectangle ++ (8,1);
\foreach \x in {1,...,7}{
    \draw (\x-0.5,-0.5) -- (\x-0.5,0.5);
}
\node at (-1,0) {$A$};
\node at ( 0,0) {12};
\node at ( 1,0) {12};
\node at ( 2,0) {11};
\node at ( 3,0) {\textcolor{red}{\textbf{12}}};
\node at ( 4,0) {\textcolor{red}{\textbf{13}}};
\node at ( 5,0) {\textcolor{red}{\textbf{13}}};
\node at ( 6,0) {10};
\node at ( 7,0) {11};

\draw[color=red] (3.55,-1.3) rectangle ++ (1.9,0.7);
\node at ( 4.5,-1) {\textcolor{red}{$\mathbf{+1}$}};

\draw[color=red] (2.55,-1.3) rectangle ++ (0.9,0.7);
\node at ( 3.0,-1) {\textcolor{red}{$\mathbf{+1}$}};

\end{tikzpicture}